\documentclass[runningheads]{llncs}
\usepackage[usenames,dvipsnames]{xcolor}
\usepackage{graphicx}
\graphicspath{{./figures/}}
\usepackage{amsmath, amssymb, amsthm}
\usepackage{thmtools, thm-restate}
\usepackage{bm}
\usepackage{mathtools}
\usepackage{stmaryrd}
\usepackage[font=scriptsize,labelfont=bf,skip=5pt]{caption}
\usepackage[disable]{todonotes}
\usepackage[inline]{enumitem}
\usepackage{url}
\usepackage{xfrac}
\usepackage{csquotes}
\usepackage{xspace}
\usepackage{listings}
\usepackage{wrapfig}
\usepackage{tikz-cd}
\usepackage{makecell}
\usepackage{apptools}
\usepackage{booktabs}
\usepackage{adjustbox}
\usepackage{multirow}
\usepackage{thm-restate}
\usepackage{nicefrac}
\usepackage{marvosym}
\usepackage{fontawesome}

\AtAppendix{\counterwithin{lemma}{section}}

\usepackage[firstpage]{draftwatermark}   

\SetWatermarkAngle{0}

\SetWatermarkText{\raisebox{15cm}{%
		\hspace{0.13cm}%
		\href{https://doi.org/10.5281/zenodo.6511363}{\includegraphics{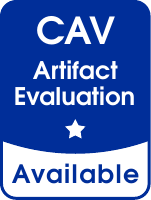}}%
		\hspace{9cm}%
		\includegraphics{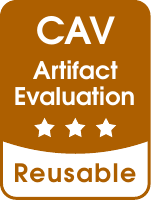}%
}}

\makeatletter
\def\thanks#1{\protected@xdef\@thanks{\@thanks
		\protect\footnotetext{#1}}}
\makeatother

\makeatletter
\renewcommand\paragraph{\@startsection{paragraph}{4}{\z@}%
	{-8\p@ \@plus -4\p@ \@minus -4\p@}%
	{-0.5em \@plus -0.22em \@minus -0.1em}%
	{\normalfont\normalsize\itshape}}
\makeatother

\makeatletter
\RequirePackage[bookmarks,unicode,colorlinks=true]{hyperref}%
\def\@citecolor{blue}%
\def\@urlcolor{blue}%
\def\@linkcolor{RedViolet}%

\def\orcidID#1{\smash{\href{http://orcid.org/#1}{\protect\raisebox{-1.25pt}{\protect\includegraphics{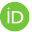}}}}}
\makeatother

\usepackage[capitalize]{cleveref}

\usepackage{footnote}
\makesavenoteenv{figure}
\makesavenoteenv{table}

\crefname{lstinputlisting}{program}{programs}
\Crefname{lstinputlisting}{Program}{Programs}

\crefname{prob}{problem}{problems}
\Crefname{prob}{Problem}{Problems}
\creflabelformat{prob}{#2\textup{(#1)}#3}

\crefname{cha}{challenge}{challenges}
\Crefname{cha}{Challenge}{Challenges}

\crefname{equation}{eq.}{eqs.}
\Crefname{equation}{Eq.}{Eqs.}
\Crefname{appendix}{Appx.}{Appxs.}
\Crefname{section}{Sect.}{Sects.}
\Crefname{figure}{Fig.}{Figs.}

\newcommand{\twk}[1]{\todo[color=blue!30,size=\scriptsize,fancyline,author=Tobias]{#1}}

\newcommand{\codify}[1]{\textup{\texttt{#1}}}
\newcommand{\while}{\codify{while}\xspace}
\newcommand{\ASSIGN}[2]{{#1} \coloneqq {#2}}
\newcommand{\WHILE}[1]{\codify{while}\:(#1)\:\{}
\newcommand{\WHILEDO}[2]{\codify{while}\:(#1)\:\{#2\}}
\newcommand{\pgcl}{\codify{pGCL}\xspace}
\newcommand{\pplname}{\codify{ReDiP}\xspace}
\newcommand{\pskip}{\codify{skip}\xspace}
\newcommand{\pcase}{\codify{case}\xspace}
\newcommand{\pbreak}{\codify{break}\xspace}

\newcommand{\pif}{\codify{if}}
\newcommand{\ITE}[3]{\codify{if}\:(#1)\:\{#2\} \: \codify{else} \: \{#3\}}
\newcommand{\IF}[1]{\codify{if}\:(#1)\:}
\newcommand{\ELSE}{\codify{else}\:}
\newcommand{\PCHOICE}[3]{\{\,#1\,\} \: [#2] \: \{\,#3\,\}}
\newcommand{\compose}[2]{{#1}{\fatsemi} \: {#2}}

\newcommand{\pelse}{\codify{else}}
\newcommand{\sem}[1]{\ensuremath{\left\llbracket #1 \right\rrbracket}}

\newcommand{\abs}[1]{\ensuremath{\vert{#1}\vert}}

\newcommand{\restrict}[2]{\ensuremath{#1_{#2}}}

\newcommand{\blue}[1]{\textcolor{RoyalBlue}{#1}}
\newcommand{\orange}[1]{\textcolor{RedOrange}{#1}}
\definecolor{webgreen}{rgb}{0,.5,0}
\newcommand{\green}[1]{\textcolor{webgreen}{#1}}
\newcommand{\annocolor}[1]{\textcolor{blue!50!green}{#1}}
\newcommand{\commentcolor}[1]{\textcolor{gray!90}{#1}}
\newcommand{\annotate}[1]{{\annocolor{\!\!{\fatslash}\!\!{\fatslash}~~\vphantom{G'} {#1}}}}
\newcommand{\comment}[1]{{\commentcolor{\texttt{/*~#1~*/}}}}

\newcommand{\bvec}[1]{\ensuremath{\mathbf{#1}}}
\renewcommand{\vec}[1]{\bvec{#1}}

\newcommand{\tto}{~{}\to{}~}
\newcommand{\wpsymbol}{\textsf{wp}}
\newcommand{\ccoloneqq}{~{}\coloneqq{}~}

\newcommand{\pgf}{\textnormal{\textsf{PGF}}\xspace}

\newcommand{\N}{\ensuremath{\mathbb{N}}\xspace}
\newcommand{\R}{\ensuremath{\mathbb{R}}\xspace}

\newcommand{\lfp}{\ensuremath{\textnormal{\textsf{lfp}}~}\xspace}

\newcommand{\eeq}{~{}={}~}

\newcommand{\ssqsubseteq}{~{}\sqsubseteq{}~}
\newcommand{\ssqsupseteq}{~{}\sqsupseteq{}~}

\newcommand{\pplus}{~{}+{}~}

\newcommand{\iin}{~{}\in{}~}

\newcommand{\mmapsto}{~{}\mapsto{}~}

\newcommand{\qiff}{\quad\textnormal{iff}\quad}

\newcommand{\qand}{\quad\textnormal{and}\quad}
\newcommand{\qqand}{\qquad\textnormal{and}\qquad}

\newcommand{\qimplies}{\quad\textnormal{implies}\quad}

\newcommand{\qqimplies}{\qquad\textnormal{implies}\qquad}

\newcommand{\expected}[1]{\mathit{E}[#1]}
\newcommand{\variance}[1]{\mathit{Var}[#1]}

\newcommand{\fsum}{\sum} 
\newcommand{\fpsring}[1]{\R[[{#1}]]}

\newcommand{\subsFPSVarFor}[3]{{#1}[{#2}/{#3}]}

\newcommand{\toolname}{\textsc{Prodigy}} 

\newcommand{\progvar}[1]{\mathtt{#1}}
\newcommand{\asgn}[2]{{#1} \,:=\, {#2}}
\newcommand{\incrasgn}[2]{{#1} ~ {+}{=} ~ {#2}}
\newcommand{\decr}[1]{#1\mathtt{--}}
\newcommand{\leftshift}[1]{\mathit{shift}^{\leftarrow}(#1)}
\newcommand{\iid}[2]{\mathtt{iid}({#1}, \, {#2})}

\newcommand{\dirac}[1]{\mathtt{dirac}(#1)}
\newcommand{\uniform}[1]{\mathtt{unif}(#1)}
\newcommand{\binomial}[2]{\mathtt{binomial}({#1},\,{#2})}
\newcommand{\nbinomial}[2]{\mathtt{nbinomial}({#1},\,{#2})}
\newcommand{\bernoulli}[1]{\mathtt{bernoulli}({#1})}

\newcommand{\geometric}[1]{\mathtt{geometric}({#1})}
\newcommand{\catalan}[1]{\mathtt{catalan}({#1})}

\newcommand{\pgfdom}{\pgf}
\newcommand{\sopdom}{\mathsf{SOP}}
\newcommand{\vars}{\mathit{Vars}}
\newcommand{\chain}{\Gamma}

\newcommand{\guard}{\varphi}
\newcommand{\charfun}[2]{\Psi_{{#1}, {#2}}}

\def\triangleforqed{\hbox{$\lhd$}}
\makeatletter
\DeclareRobustCommand{\qedTT}{%
	\ifmmode
	\eqno \def\@badmath{$$}
	\let\eqno\relax \let\leqno\relax \let\veqno\relax
	\hbox{\triangleforqed}%
	\else
	\leavevmode\unskip\penalty9999 \hbox{}\nobreak\hfill
	\quad\hbox{\triangleforqed}%
	\fi
}
\makeatother

\newif\ifcameraready
\camerareadyfalse

\newif\ifblindreview
\blindreviewfalse

\begin{document}
\title{
    Does a Program Yield the Right Distribution?%
    \ifblindreview\else\thanks{\setlength{\leftskip}{0em}%
    	This research was funded by the ERC Advanced Project FRAPPANT under grant No.\ 787914, by the 
    	EU's Horizon 2020 research and innovation programme under the Marie Sk\l{}odowska-Curie grant 
    	No.~101008233, 
    	and by the DFG RTG 2236 UnRAVeL.}
    \fi
}
\subtitle{Verifying Probabilistic Programs via Generating Functions}

\titlerunning{Does a Program Yield the Right Distribution?}
%

\authorrunning{%
	\ifblindreview%
	\else%
	M.\ Chen, J.-P.\ Katoen, L.\ Klinkenberg, T.\ Winkler
	\fi
}

\ifblindreview%
\author{}
\institute{}
\else%
\author{Mingshuai Chen$^{\text{(\Letter)}}$\orcidID{0000-0001-9663-7441} \and
Joost-Pieter Katoen$^{\text{(\Letter)}}$\orcidID{0000-0002-6143-1926} \and\\
Lutz Klinkenberg$^{\text{(\Letter)}}$\orcidID{0000-0002-3812-0572} \and
Tobias Winkler$^{\text{(\Letter)}}$\orcidID{0000-0003-1084-6408}}
\institute{RWTH Aachen University, Aachen, Germany \\
\email{\{chenms,katoen,lutz.klinkenberg,tobias.winkler\}@cs.rwth-aachen.de}}
\fi

\maketitle              
%

\setlength{\floatsep}{.5\baselineskip}
\setlength{\textfloatsep}{.8\baselineskip}
\setlength{\intextsep}{.7\baselineskip}

\begin{abstract}
    We study discrete probabilistic programs with potentially unbounded looping behaviors over an infinite state space.
    We present, to the best of our knowledge, \emph{the first decidability result for the problem of determining whether such a program generates exactly a specified distribution over its outputs} (provided the program terminates almost-surely). The class of distributions that can be specified in our formalism consists of standard distributions (geometric, uniform, etc.) and finite convolutions thereof. Our method relies on representing these (possibly infinite-support) distributions as \emph{probability generating functions} which admit effective arithmetic operations.
    We have automated our techniques in a tool called $\toolname$, which supports automatic invariance checking, compositional reasoning of nested loops, and efficient queries 
    to the output distribution, as demonstrated by experiments.
\end{abstract}
\section{Introduction}\label{sec:intro}

Probabilistic programs~\cite{DBLP:journals/jcss/Kozen81,ACM:conf/fose/Gordon14,DBLP:journals/corr/abs-1809-10756} augment deterministic programs with stochastic behaviors, e.g., random sampling, probabilistic choice, and conditioning (via posterior observations).
Probabilistic programs have undergone a recent surge of interest due to prominent applications in a wide range of domains: they steer autonomous robots and self-driving cars~\cite{agentmodels,DBLP:conf/mfi/ShamsiFGN20}, are key to describe security~\cite{DBLP:journals/toplas/BartheKOB13} and quantum~\cite{DBLP:journals/toplas/Ying11} mechanisms, intrinsically code up randomized algorithms for solving NP-hard or even deterministically unsolvable problems (in, e.g., distributed computing~\cite{DBLP:journals/csur/Schneider93,DBLP:journals/jal/AspnesH90}), and are rapidly encroaching on AI as well as approximate computing~\cite{DBLP:journals/cacm/CarbinMR16}. See~\cite{barthe_katoen_silva_2020} for recent advancements in probabilistic programming.

The crux of probabilistic programming, \`{a} la Hicks' interpretation~\cite{Hicks-blog2014}, is to \emph{treat normal-looking programs as if they were probability distributions}. A random-number generator, for instance, is a probabilistic program that produces a uniform distribution across numbers from a range of interest.
Such a lift from deterministic program states to possibly infinite-support distributions (over states) renders the verification problem of probabilistic programs notoriously hard~\cite{DBLP:journals/acta/KaminskiKM19}.
In particular, reasoning about probabilistic loops often amounts to computing quantitative fixed-points which are highly intractable in practice.
As a consequence, existing techniques are mostly concerned with approximations, i.e., they strive for verifying or obtaining upper and/or lower bounds on various quantities like assertion-violation probabilities~\cite{DBLP:conf/pldi/WangS0CG21}, preexpectations~\cite{DBLP:conf/cav/BatzCKKMS20,DBLP:journals/pacmpl/HarkKGK20}, moments~\cite{DBLP:conf/pldi/Wang0R21}, expected runtimes~\cite{DBLP:journals/jacm/KaminskiKMO18}, and concentrations~\cite{DBLP:conf/cav/ChakarovS13,DBLP:conf/cav/ChatterjeeFG16}, which reveal only partial information about the probability distribution carried by the program.

In this paper, we address the problem of \emph{how to determine whether a (possibly infinite-state) probabilistic program yields exactly the desired (possibly infinite-support) distribution under all possible inputs}.
We highlight two scenarios where encoding the \emph{exact} distribution -- other than (bounds on) the above-mentioned quantities -- is of particular interest:
\begin{enumerate*}[label=(\Roman*)]
	\item In many safety- and/or security-critical domains, e.g., cryptography, a slightly perturbed distribution (while many of its probabilistic quantities remain unchanged) may lead to significant attack vulnerabilities or even complete compromise of the cryptographic system, see, e.g., Bleichenbacher's biased-nonces attack~\cite[Sect.~5.10]{Heninger_2021} against the probabilistic Digital Signature Algorithm. Therefore, the system designer has to impose a complete specification of the anticipated distribution produced by the probabilistic component.
	\item In the context of quantitative verification, the user may be interested in multiple properties (of different types, e.g., the aforementioned quantities) of the output distribution carried by a probabilistic program.
	In absence of the exact distribution, multiple analysis techniques -- tailored to different types of properties -- have to be applied in order to answer all queries from the user.
\end{enumerate*}
We further motivate our problem using a concrete example as follows.

\begin{example}[Photorealistic Rendering~\textnormal{\cite{KajiyaJ86}}]
    \label{exmp:russian-roulette}
    Monte Carlo integration algorithms form a well-known class of probabilistic programs which approximate complex integral expressions by sampling~\cite{hammersley2013monte}.
    One of its particular use-cases is the photorealistic rendering of virtual scenes by a technique called \emph{Monte Carlo path tracing} (MCPT)~\cite{KajiyaJ86}.
    
    MCPT works as follows:
    For every pixel of the output image, it shoots $n$ sample rays into the scene and models the light transport behavior to approximate the incoming light at that particular point.
    Starting from a certain pixel position, MCPT randomly chooses a direction, traces it until a scene object is hit, and then proceeds by either 
    \begin{enumerate*}[label=(\roman*)]
    	\item terminating the tracing and evaluating the overall ray, or
    	\item continuing the tracing by computing a new direction.
    \end{enumerate*}
    In the physical world, the light ray may be reflected arbitrarily often and thus stopping the tracing after a certain 
    amount of bounces would introduce a bias in the integral estimation.
    As a remedy, the decision when to stop the tracing is made in a \emph{Russian roulette} manner by flipping a coin\footnote{The bias of the coin depends on the material's \emph{reflectivity}: a reflecting material such as a mirror requires more light bounces than an absorptive one, e.g., a black surface.} at each intersection point~\cite{SIGGRAPH:ArvoJ90}.
    
    The program in \Cref{fig:russianroulette} is an implementation of a simplified MCPT path generator. 
    The cumulative length of all $\progvar{n}$ rays is stored in the (random) variable $\progvar{c}$, which is directly proportional to MCPT's expected runtime.
    The implementation is designed in a way that \emph{$\progvar{c}$ induces a distribution as the sum of $\progvar{n}$ independent and identically distributed (i.i.d.) geometric random variables} such that the resulting integral estimation is unbiased. In our framework, we view such an exact output distribution of $\progvar{c}$ as a \emph{specification} and verify -- fully automatically -- that the implementation in \Cref{fig:russianroulette} with nested loops indeed satisfies this specification. 
    \qedTT
\end{example}

\begin{figure}[t]
    \begin{align*}
    & \WHILE{\progvar{n} > 0} ~ \comment{generate n samples}\\
    & \qquad \ASSIGN{\progvar{running}}{1}\fatsemi\\
    & \qquad \WHILE{\progvar{running} = 1} ~ \comment{generate a light ray}\\
    &\qquad \qquad \PCHOICE{\asgn{\progvar{running}}{0}~\comment{absorb}}{\nicefrac 1 2}{\asgn{\progvar{c}}{\progvar{c}+1}~\comment{reflect}} ~ \} \fatsemi \\
    & \qquad \ASSIGN{\progvar{n}}{\progvar{n} - 1}~\}
    \end{align*}
    \caption{
        Monte Carlo path tracing in a scene with constant reflectivity $\nicefrac 1 2$.
    }
    \label{fig:russianroulette}
\end{figure}

\paragraph*{Approach.}
Given a probabilistic loop $L = \WHILEDO{\guard}{P}$ with guard $\guard$ and loop-free body $P$, we aim to determine whether $L$ agrees with a specification $S$:
\begin{align}\label[prob]{prob:agreement}
	L = \WHILEDO{\guard}{P} \quad \stackrel{?}{\sim} \quad S~, \tag{$\star$}
\end{align}%
namely, whether $L$ yields -- upon termination -- exactly the same distribution as encoded by $S$ under all possible program inputs. This problem is non-trivial:
\begin{enumerate*}[label=(C\arabic*)]
	\item\label[cha]{challenge-1}%
	$L$ may induce an infinite state space and infinite-support distributions, thus making techniques like probabilistic bounded model checking~\cite{DBLP:conf/atva/0001DKKW16} insufficient for verifying the property by means of unfolding the loop $L$.
	\item\label[cha]{challenge-2}%
	There is, to the best of our knowledge, a lack of non-trivial characterizations of $L$ and $S$ such that \cref{prob:agreement} admits a decidability result.
	\item\label[cha]{challenge-3}%
	To decide \cref{prob:agreement} -- even for a loop-free program $L$ -- one has to account for infinitely or even uncountably many inputs such that $L$ yields the same distribution as encoded by $S$ when being deployed in all possible contexts.
\end{enumerate*}

We address \cref{challenge-1} by exploiting the forward denotational semantics of probabilistic programs based on \emph{probability generating function} (PGF) representations of (sub-)distributions~\cite{lopstr/klinkenberg20}, which benefits crucially from closed-form (i.e., finite) PGF representations of possibly infinite-support distributions.
A probabilistic program $L$ hence acts as a transformer $\sem{L}(\cdot)$ that transforms an input PGF $g$ into an output PGF $\sem{L}(g)$ (as an instantiation of Kozen's transformer semantics~\cite{DBLP:journals/jcss/Kozen81}).
In particular, we \emph{interpret the specification $S$ as a loop-free probabilistic program $I$}.
Such an identification of specifications with programs has two important advantages:
\begin{enumerate*}[label=(\roman*)]
	\item we only need a single language to encode programs as well as specifications, and
	\item it enables compositional reasoning in a straightforward manner, in particular, the treatment of nested loops.
\end{enumerate*}
The problem of checking $L \sim S$ then boils down to checking whether $L$ and $I$ transform every possible input PGF into the same output PGF:
\begin{align}\label[prob]{prob:loop-equiv}
	\forall g \in \pgfdom\colon\quad \llbracket\underbrace{\WHILEDO{\varphi}{P}}_{L}\rrbracket (g) \quad \stackrel{?}{=} \quad \sem{I}(g)~. \tag{$\dagger$}
\end{align}%
As $I$ is loop free, \cref{prob:loop-equiv} can be reduced 
to checking the equivalence of two \emph{loop-free} probabilistic programs (cf.\ \cref{thm:fpToEquiv}):
\begin{align}\label[prob]{prob:loop-free-equiv}
	\forall g \in \pgfdom\colon\quad \sem{\ITE{\guard}{P \fatsemi I}{\pskip}}(g) \quad \stackrel{?}{=} \quad \sem{I}(g)~. \tag{$\ddagger$}
\end{align}%
Now \cref{challenge-3} applies since the universal quantification in \cref{prob:loop-free-equiv} requires to determine the equivalence against infinitely many -- possibly infinite-support -- distributions over program states. We facilitate such an equivalence checking by developing a \emph{second-order PGF} (SOP) semantics for probabilistic programs, which naturally extends the PGF semantics while allowing to reason about infinitely many PGF transformations simultaneously (see \cref{thm:equivCharac}).

Finally, to obtain a decidability result (cf.\ \cref{challenge-2}), we develop the \emph{rectangular discrete probabilistic programming language} ($\pplname$) -- a variant of $\pgcl$~\cite{McIver04} with syntactic restrictions to rectangular guards -- featuring various nice properties, e.g., they inherently support i.i.d.\ sampling, and in particular, they \emph{preserve closed-form PGF} when acting as PGF transformers. We show that \emph{\cref{prob:loop-free-equiv} is decidable for $\pplname$ programs $P$ and $I$ if all the distribution statements therein have rational closed-form PGF}
(cf.\ \cref{thm:loop-free-decidability}).
As a consequence, \emph{\cref{prob:loop-equiv} and thereby \cref{prob:agreement} of checking $L \sim S$ are decidable if $L$ terminates almost-surely on all possible inputs $g$} (cf.\ \cref{cor:loop-equiv-decidability}).

\paragraph*{Demonstration.}
We have automated our techniques in a tool called $\toolname$. As an example, $\toolname$ was able to verify, fully automatically in 25 milliseconds, that the implementation of the MCPT path generator with nested loops (in \Cref{fig:russianroulette}) is indeed equivalent to the loop-free program
\[
\incrasgn{\progvar{c}}{\iid{\geometric{\nicefrac{1}{2}}}{\progvar{n}}} \fatsemi \, \asgn{\progvar{n}}{0}
\]%
which encodes the specification that, upon termination, $\progvar{c}$ is distributed as the sum of $\progvar{n}$ i.i.d.\ geometric random variables.
%
%
With such an output distribution, multiple queries can be efficiently answered by applying standard PGF operations.
For example, the expected value and variance of the runtime are 
$\expected{\progvar{c}} = n$ and $\variance{\progvar{c}} = 2n$, respectively (assuming $\progvar{c}=0$ initially).

\paragraph{\bf Contributions.}
The main contributions of this paper are:
\begin{itemize}
	\item The probabilistic programming language $\pplname$ and its forward denotational semantics as PGF transformers.
	We show that loop-free $\pplname$ programs preserve closed-form PGF.
	\item The notion of SOP that enables reasoning about infinitely many PGF transformations simultaneously.
	We show that the problem of determining whether an infinite-state $\pplname$ loop generates -- upon termination -- exactly a specified distribution is decidable.
	\item The software tool $\toolname$ which supports automatic invariance checking on the source-code level; 
	it allows reasoning about nested $\pplname$ loops in a compositional manner, and supports efficient queries on various quantities including assertion-violation probabilities, expected values, (high-order) moments, precise tail probabilities, as well as concentration bounds.
\end{itemize}

\paragraph{Organization.}
We introduce generating functions in \cref{sec:prelim} and define the $\pplname$ language in \cref{sec:ppl}. \cref{sec:pgfsem} presents the PGF semantics. \cref{sec:loops} establishes our decidability result in reasoning about $\pplname$ loops, with case studies in \cref{sec:case-studies}. After discussing related work in \cref{sec:related-work}, we conclude the paper in \cref{sec:conclusion}. Further details, e.g, proofs and additional examples, can be found in
\ifcameraready
the full version~\cite{DBLP:journals/corr/abs-2205-01449}.
\else
the appendices.
\fi
\section{Generating Functions}
\label{sec:prelim}

\begin{quote}
    \emph{``A generating function is a clothesline on which we hang up a sequence of numbers for display.''}
    \raggedleft --- H.\ S.\ Wilf, Generatingfunctionology~\cite{generatingfunctionology}
\end{quote}
The method of \emph{generating functions} (GF) is a vital tool in many areas of mathematics.
This includes in particular enumerative combinatorics~\cite{DBLP:books/daglib/0023751,generatingfunctionology} and -- most relevant for this paper -- probability theory~\cite{JKK1993}.
In the latter, the sequences \enquote{hanging on the clotheslines} happen to describe probability distributions over the non-negative integers $\N$, e.g., $\nicefrac{1}{2}, \nicefrac{1}{4}, \nicefrac{1}{8}, \ldots$ (aka, the geometric distribution).

The most common way to relate an (infinite) \emph{sequence} of numbers to a generating \emph{function} relies on the familiar Taylor series expansion:
Given a sequence, for example $\nicefrac{1}{2}, \nicefrac{1}{4}, \nicefrac{1}{8}, \ldots$, find a function $x \mapsto f(x)$ whose Taylor series around $x=0$ uses the numbers in the sequence as coefficients.
In our example, 
\begin{equation}
    \label{eq:gfIntro}
    \frac{1}{2 - x} \eeq \frac{1}{2} + \frac{1}{4} x + \frac{1}{8} x^2 + \frac{1}{16} x^3 + \frac{1}{32} x^4 + \ldots ~,
\end{equation}
for all $|x| < 2$, hence the \enquote{clothesline} used for hanging up $\nicefrac{1}{2}, \nicefrac{1}{4}, \nicefrac{1}{8}, \ldots$ is the function $1/(2-x)$.
Note that the GF is a -- from a purely syntactical point of view -- \emph{finite} object while the sequence it represents is \emph{infinite}.
A key strength of this technique is that many meaningful operations on infinite series can be performed by manipulating an encoding GF (see \Cref{fig:cheatSheet} for an overview and examples).
In other words, GF provide an \emph{interface} to perform operations on and extract information from infinite sequences in an effective manner.

\subsection{The Ring of Formal Power Series}

Towards our goal of encoding distributions over \emph{program states} (valuations of finitely many integer variables) as generating functions, we need to consider \emph{multivariate} GF, i.e., GF with more than one variable.
Such functions represent multidimensional sequences, or \emph{arrays}.
%
Since multidimensional Taylor series quickly become unhandy, we will follow a more \emph{algebraic} approach that is also advocated in~\cite{generatingfunctionology}:
We treat sequences and arrays as elements from an algebraic structure: the \emph{ring of Formal Power Series} (FPS).
Recall that a (commutative) \emph{ring} $(A,+,\cdot,0,1)$ consists of a non-empty carrier set $A$, associative and commutative binary operations \enquote{$+$} (addition) and \enquote{$\cdot$} (multiplication) such that multiplication distributes over addition, and neutral elements $0$ and $1$ w.r.t.\ addition and multiplication, respectively.
Further, every $a \in A$ has an additive inverse $-a \in A$.
Multiplicative inverses $a^{-1} = 1/a$ need not always exist.
Let $ k \in \N = \{0,1,\ldots\}$ be fixed in the remainder.

\begin{table}[t]
	\centering
	\caption{
		GF cheat sheet.
		$f, g$ and $X, Y$ are arbitrary GF and indeterminates, resp.
	}
	\label{fig:cheatSheet}
	\begin{adjustbox}{max width=1.0\textwidth}
		\renewcommand{\arraystretch}{1.6}
		\setlength{\tabcolsep}{5pt}
		\begin{tabular}{l l l}
			\toprule
			\textit{Operation} & \textit{Effect}& \textit{(Running) Example}\\ \midrule
			$f^{-1} = 1/f$ & \makecell[l]{Multiplicative inverse of $f$ \\[-2pt] (if it exists)} & \makecell[l]{$\frac{1}{1 - XY} = 1 + XY + X^2Y^2 + \ldots$ \\[-1pt] because $(1- XY)(1 + XY + X^2Y^2 + \ldots) = 1$} \\ 
			$fX$ & Shift in dimension $X$ & $\frac{X}{1 - XY} = X + X^2Y + X^3Y^2 + \ldots$ \\ 
			$\subsFPSVarFor{f}{X}{0}$ & Drop terms containing $X$ & $\frac{1}{1 - 0Y} = 1$\\ 
			$\subsFPSVarFor{f}{X}{1}$ & Projection\footnote{%
				Projections are not always well-defined, e.g., $\subsFPSVarFor{\frac{1}{1-X+Y}}{X}{1} = \frac{1}{Y}$ is ill-defined because $Y$ is not invertible.
				However, in all situations where we use projection it will be well-defined; in particular, projection is well-defined for PGF.
			} on $Y$ & $\frac{1}{1 - 1Y} = 1 + Y + Y^2 + \ldots$ \\ 
			$fg$ & \makecell[l]{Discrete convolution \\[-2pt] (or Cauchy product)}& $\frac{1}{(1 - XY)^2} = 1 + 2XY + 3X^2Y^2 + \ldots$ \\ 
			$\partial_X f$ & Formal derivative in $X$ & $\partial_X \frac{1}{1 - XY} = \frac{Y}{(1-XY)^2} = Y +2XY^2 + 3X^2Y^3 + \ldots$ \\ 
			$f + g$ & Coefficient-wise sum & $\frac{1}{1 - XY} {+} \frac{1}{(1 - XY)^2} = \frac{2-XY}{(1-XY)^2} = 2 {+} 3XY {+} 4 X^2Y^2 {+} \ldots$ \\ 
			$a f$ & Coefficient-wise scaling & $\frac{7}{(1 - XY)^2} = 7 + 14XY + 21 X^2Y^2 + \ldots$ \\
			%
			%
			\bottomrule
		\end{tabular}
	\end{adjustbox}%
\end{table}

\begin{definition}[The Ring of FPS]
    \label{def:fps}
    A $k$-dimensional FPS is a $k$-dim.\ array $f \colon \N^k \to \R$.
    We denote FPS as \emph{formal sums} as follows:
    Let $\vec{X} {=} (X_1,\ldots, X_k)$ be an ordered vector of symbols, called \emph{indeterminates}.
    The FPS $f$ is written as
    \[
        f \eeq \fsum\nolimits_{\sigma \in \N^k} f(\sigma) \vec{X}^\sigma
    \]
    where $\vec{X}^\sigma$ is the \emph{monomial} $X_1^{\sigma_1} X_2^{\sigma_2} \cdots X_k^{\sigma_k}$.    
    The \emph{ring of FPS} is denoted $\fpsring{\vec{X}}$ where the operations are defined as follows:
    For all $f,g \in \fpsring{\vec{X}}$ and $\sigma \in \N^k$,
    $(f + g)(\sigma) = f(\sigma) + g(\sigma)$, and $(f \cdot g)(\sigma) = \sum_{\sigma_1+\sigma_2=\sigma}f(\sigma_1)g(\sigma_2)$.
    
\end{definition}
The multiplication $f \cdot g$ is the usual \emph{Cauchy product} of power series (aka discrete convolution); it is well defined because for all $\sigma \in \N^k$ there are just \emph{finitely} many $\sigma_1 + \sigma_2 = \sigma$ in $\N^k$.
We write $fg$ instead of $f \cdot g$.

The formal sum notation is standard in the literature and often useful because the arithmetic FPS operations are very similar to how one would do calculations with \enquote{real} sums.
We stress that the indeterminates $\vec{X}$ are merely \emph{labels} for the $k$ dimensions of $f$ and do not have any other particular meaning.
In the context of this paper, however, it is natural to identify the indeterminates with the program variables (e.g. indeterminate $X$ refers to variable $\progvar{x}$, see \Cref{sec:ppl}).

\Cref{eq:gfIntro} can be interpreted as follows in the ring of FPS:
The \enquote{sequences} $2 - 1X + 0X^2 + \ldots$ and $\nicefrac{1}{2} + \nicefrac{1}{4}X + \nicefrac{1}{8}X^2 + \ldots$ are (multiplicative) \emph{inverse} elements to each other in $\fpsring{X}$, i.e., their product is $1$.
More generally, we say that an FPS $f$ is \emph{rational} if $f = gh^{-1} = g/h$ where $g$ and $h$ are polynomials, i.e., they have at most finitely many non-zero coefficients;
and we call such a representation a \emph{rational closed form}.

\ifcameraready
A more extensive introduction to FPS can be found in~\cite[Appx.~D]{DBLP:journals/corr/abs-2205-01449}.
\else
We include a more extensive self-contained introduction to FPS in \Cref{sec:fpsbackground}.
\fi

\subsection{Probability Generating Functions}

We are especially interested in GF that describe probability distributions.

\begin{definition}[PGF]
   \label{def:pgfmaintext}
   A $k$-dimensional FPS $g$ is a \emph{probability generating function} (PGF) if \emph{(i)} for all $\sigma \in \N^k$ we have $g(\sigma) \geq 0$, and \emph{(ii)} $\sum_{\sigma \in \N^k} g(\sigma) \leq 1$.
\end{definition}
For example, \eqref{eq:gfIntro} is the PGF of a $\nicefrac{1}{2}$-geometric distribution.
The PGF of other standard distributions are given in \Cref{fig:distExp} further below.
Note that \Cref{def:pgfmaintext} also includes \emph{sub-PGF} where the sum in (ii) is strictly less than $1$.


\section{$\pplname$: A Probabilistic Programming Language}
\label{sec:ppl}

This section presents our 
\emph{Rectangular Discrete Probabilistic Programming Language}, or $\pplname$ for short.
The word \enquote{rectangular} refers to a restriction we impose on the guards of conditionals and loops, see \Cref{sec:syntax}.
$\pplname$ is a variant of $\pgcl$~\cite{McIver04} with some extra syntax but also some syntactic restrictions.

\subsection{Program States and Variables}
\label{sec:statesAndVars}

Every $\pplname$-program $P$ operates on a finite set of $\N$-valued \emph{program variables} $\vars(P) =  \{\progvar{x}_1,\ldots,\progvar{x}_k\}$.
We do not consider negative or non-integer variables.
A \emph{program state} of $P$ is thus a mapping $\sigma \colon \vars(P) \to \N$.
As explained in \Cref{sec:intro}, the key idea is to represent distributions over such program states as PGF.
Consequently, we identify a single program state $\sigma$ with the \emph{monomial} $\vec{X}^\sigma = X_1^{\sigma(\progvar{x}_1)} \cdots X_k^{\sigma(\progvar{x}_k)}$ where $X_1,\ldots,X_k$ are indeterminates representing the program variables $\progvar{x}_1,\ldots,\progvar{x}_k$.
We will stick to this notation: throughout the whole paper, we typeset program variables as $\progvar{x}$ and the corresponding FPS indeterminate as $X$.
The initial program state on which a given $\pplname$-program is supposed to operate must always be stated explicitly.

\subsection{Syntax of $\pplname$}
\label{sec:syntax}

The syntax of $\pplname$ is defined inductively, see the leftmost column of \Cref{fig:synsem}.
Here, $\progvar{x}$ and $\progvar{y}$ are program variables, $n \in \N$ is a constant, $D$ is a \emph{distribution expression} (see \Cref{fig:distExp}), and $P_1, P_2$ are $\pplname$-programs.
The general idea of $\pplname$ is to provide a minimal core language to keep the theory simple.
Many other common language constructs such as linear arithmetic updates $\asgn{\progvar{x}}{2\progvar{y} + 3}$ are expressible in this core language. See~\ifcameraready\cite[Appx.~A]{DBLP:journals/corr/abs-2205-01449} \else\Cref{sec:syntsugar} \fi for a complete specification.

\newcommand{\synsemtab}{\renewcommand{\arraystretch}{1.6}
    \setlength{\tabcolsep}{5pt}
    \begin{adjustbox}{max width=\textwidth}
        \begin{tabular}{l l l}
            \toprule
            \textit{$\pplname$-program $P$} & \textit{Semantics $\sem{P}(g)$ -- see \cref{sec:explainTransformer}} & \textit{Description} \\
            \midrule
            $\asgn{\progvar{\blue{x}}}{\orange{n}}$& $ \subsFPSVarFor{g}{\blue{X}}{1} \blue{X}^{\orange{n}}$ & Assign const.\ $\orange{n} \in \N$ to var.\ $\progvar{\blue{x}}$\\
            $\decr{\progvar{\blue{x}}}$& $ (g - \subsFPSVarFor{g}{\blue{X}}{0}) \blue{X}^{- 1} \pplus \subsFPSVarFor{g}{\blue{X}}{0}$ & Decr.\ $\progvar{\blue{x}}$ (\enquote{monus} semantics)\\
            $\incrasgn{\progvar{\blue{x}}}{\iid{D}{\progvar{\green{y}}}}$ & $\subsFPSVarFor{g}{\green{Y}}{\green{Y} \subsFPSVarFor{\sem{D}}{T}{\blue{X}}}$ & \makecell[l]{Incr.\ $\progvar{\blue{x}}$ by the sum of $\progvar{\green{y}}$ i.i.d. \\ samples from $D$ -- see \cref{sec:iid}} \\[5pt] 
            \makecell[l]{$\IF{\progvar{\blue{x}} < \orange{n}} \{P_1\}$ \\[0pt] $\ELSE \{P_2\}$} & \makecell[l]{$\sem{P_1} (g_{\blue{x}<\orange{n}}) \pplus \sem{P_2} \left(g - g_{\blue{x}<\orange{n}}\right)$ ~, where \\[0pt] $g_{\blue{x}<n} \eeq \sum_{i = 0}^{\orange{n}-1} \frac{1}{i!} \subsFPSVarFor{(\partial_{\blue{X}}^i g)}{\blue{X}}{0} \blue{X}^i$} & Conditional branching \\[5pt]
            $\compose{P_1}{P_2}$ & $\sem{P_2}(\sem{P_1}(g))$ & Sequential composition \\
            \midrule
            $\WHILEDO{\progvar{\blue{x}} < \orange{n}}{P_1}$& \makecell[l]{$\left[\lfp \charfun{\progvar{\blue{x}}<\orange{n}}{P_1}\right] (g)$ ~, where \\[0pt] $\charfun{\progvar{\blue{x}}<\orange{n}}{P_1}(\psi) = \lambda f .~ (f {-} \restrict{f}{\progvar{\blue{x}}<\orange{n}}) {+} \psi(\sem{P_1}(\restrict{f}{\progvar{\blue{x}}<\orange{n}}))$} & \makecell[l]{Loop defined as fixed point} \\
            \bottomrule
        \end{tabular}
\end{adjustbox}}

\begin{table}[t]
    \centering
    \caption{
    	Syntax and semantics of $\pplname$.
    	$g$ is the input PGF.
    }\label{fig:synsem}
    \synsemtab
\end{table}
\begin{table}[t]
    \setlength{\tabcolsep}{5pt}
    \caption{
    	A non-exhaustive list of common discrete distributions with rational PGF.
    	The parameters $p$, $n$, and $\lambda$ are a probability, a natural, and a non-negative real number, respectively.
    	$T$ is a reserved placeholder indeterminate.
    }
    \label{fig:distExp}
    \begin{adjustbox}{max width=\textwidth}
        \begin{tabular}{l  l l}
            \toprule
            $D$ & $\sem{D}$ & \emph{Description} \\
            \midrule
            $\dirac{n}$ & $T^n$ & Point mass \\[.5em]
            $\bernoulli{p}$ & $1-p + pT$ & Bernoulli distribution (coin flip)\\[.5em]
            $\uniform{n}$ & $ (1 - T^{n}) \,/\, n(1-T)$ & Discrete uniform distribution on $\{0,\ldots,n{-}1\}$\\[.5em]
            $\geometric{p}$ & $ (1-p) \,/\, (1 - pT)$ & Geometric distribution (no.\ trials until first success)\\[.5em]
            $\binomial{p}{n}$ & $(1-p + pT)^n$ & Binomial distribution (successes of $n$ yes-no experiments)\\[.5em]
            $\nbinomial{p}{n}$ & $(1-p)^n \,/\, (1 - pT)^n$ & Negative binomial distribution \\
            %
            \bottomrule
        \end{tabular}
    \end{adjustbox}
\end{table}

The word ``rectangular'' in $\pplname$ emphasizes that our $\mathtt{if}$-guards can only identify \emph{axis-aligned hyper-rectangles}\footnote{More precisely, we can simulate statements like $\ITE{R}{...}{...}$, where $R$ is a finite Boolean combination of rectangular guards, through appropriate nesting of $\IF$; note that such an $R$ is indeed a finite union of axis-aligned reactangles in $\N^k$.} in $\N^k$, but no more general polyhedra.
These \emph{rectangular guards} $\progvar{x} < n$ have the fundamental property that they preserve rational PGF.
On the other hand, allowing more general guards like $\progvar{x} < \progvar{y}$ breaks this property (see \cite{DBLP:conf/soda/FlajoletPS11} and our comments in~\ifcameraready\cite[Appx.~B]{DBLP:journals/corr/abs-2205-01449}\else\Cref{app:nonrect}\fi).

The most intricate feature of $\pplname$ is the -- potentially unbounded -- loop $\WHILEDO{\progvar{x} < n}{P}$.
A program that does not contain loops is called \emph{loop-free}.

\subsection{The Statement $\incrasgn{\progvar{x}}{\iid{D}{\progvar{y}}}$}
\label{sec:iid}

The novel $\mathtt{iid}$ statement is the heart of the loop-free fragment of $\pplname$ -- it subsumes both $\asgn{\progvar{x}}{D}$ (\enquote{assign a $D$-distributed sample to $\progvar{x}$}) and the standard assignment $\asgn{\progvar{x}}{\progvar{y}}$.
We include the assign-increment (${+}{=}$) version of $\mathtt{iid}$ in the core fragment of $\pplname$ for technical reasons;
the assignment $\asgn{\progvar{x}}{\iid{D}{\progvar{y}}}$ can be recovered from that as syntactic sugar by simply setting $\asgn{\progvar{x}}{0}$ beforehand.

Intuitively, the meaning of $\incrasgn{\progvar{x}}{\iid{D}{\progvar{y}}}$ is as follows.
The right-hand side $\iid{D}{\progvar{y}}$ can be seen as a function that takes the current value $v$ of variable $\progvar{y}$, then draws $v$ i.i.d.\ samples from distribution $D$, computes the sum of all these samples and finally increments $\progvar{x}$ by the so-obtained value.
For example, to perform $\asgn{\progvar{x}}{\progvar{y}}$, we may just write
$
    \asgn{\progvar{x}}{\iid{\dirac{1}}{\progvar{y}}}
$
as this will draw $\progvar{y}$ times the number $1$, then sum up these $\progvar{y}$ many 1's to obtain the result $\progvar{y}$ and assign it to $\progvar{x}$.
Similarly, to assign a random sample from a, say, uniform distribution to $\progvar{x}$, we can execute
$
    \asgn{\progvar{y}}{1} \, \fatsemi \, \asgn{\progvar{x}}{\iid{\uniform{n}}{\progvar{y}}}.
$

But $\mathtt{iid}$ is not only useful for defining standard operations.
In fact, taking sums of i.i.d.\ samples is common in probability theory.
The \emph{binomial distribution} with parameters $p \in (0,1)$ and $n \in \N$, for example, is the defined as the sum of $n$ i.i.d.\ Bernoulli-$p$-distributed samples and thus
\[
    \asgn{\progvar{x}}{\binomial{p}{\progvar{y}}}
    \qquad
    \text{is equivalent to}
    \qquad 
    \asgn{\progvar{x}}{\iid{\bernoulli{p}}{\progvar{y}}}
\]%
for all constants $p \in (0,1)$.
Similarly, the \emph{negative} $(p,n)$-binomial distribution is the sum of $n$ i.i.d.\ geometric-$p$-distributed samples.
Overall, $\mathtt{iid}$ renders the loop-free fragment of $\pplname$ \emph{strictly more expressive} than it would be if we had included only $\asgn{\progvar{x}}{D}$ and $\asgn{\progvar{x}}{\progvar{y}}$ instead.
As a consequence, since we use loop-free programs as a specification language (see \Cref{sec:loops}), $\mathtt{iid}$ enables us to write more expressive program specifications while retaining decidability.


\section{Interpreting $\pplname$ with PGF}
\label{sec:pgfsem}

In this section, we explain the PGF-based semantics of our language which is given in the second column of \Cref{fig:synsem}.
The overall idea is to view a $\pplname$-program $P$ as a \emph{distribution transformer}~\cite{DBLP:journals/jcss/Kozen85,McIver04}.
This means that the input to $P$ is a \emph{distribution} over initial program states (inputting a deterministic state is just the special case of a Dirac distribution), and the output is a distribution over final program states.
With this interpretation, if one regards distributions as \emph{generalized program states}~\cite{jacobs2019logical}, a probabilistic program is actually \emph{deterministic}:
The same input distribution always yields the same output distribution.
The goal of our PGF-based semantics is to construct an \emph{interpreter} that executes a $\pplname$-program statement-by-statement in forward direction, transforming one generalized program state into the next.
We stress that these generalized program states, or distributions, can be infinite-support in general.
For example, the program $\asgn{\progvar{x}}{\geometric{0.5}}$ outputs a geometric distribution -- which has infinite support -- on $\progvar{x}$.

\subsection{A Domain for Distribution Transformation}
\label{sec:pgfdom}

We now define a domain, i.e., an \emph{ordered} structure, where our program's in- and output distributions live.
Following the general idea of this paper, we encode them as PGF.
Let $\vars$ be a fixed finite set of program variables $\progvar{x}_1 ,\ldots , \progvar{x}_k$ and let $\vec{X} = (X_1,\ldots,X_k)$ be corresponding formal indeterminates.
We let $\pgfdom = \{g \in \R[[\vec{X}]] \mid g \text{ is a PGF}\}$ denote the set of all PGF.
Recall that this also includes sub-PGF (\Cref{def:pgfmaintext}).
Further, we equip $\pgfdom$ with the pointwise order, i.e., we let $g \sqsubseteq f$ iff $g(\sigma) \leq f(\sigma)$ for all $\sigma \in \N^k$.
It is clear that $(\pgfdom, \sqsubseteq)$ is a partial order that is moreover \emph{$\omega$-complete}, i.e., there exists a least element $0$ and all ascending chains $\chain = \{ g_0 \sqsubseteq g_1 \sqsubseteq \ldots \}$ in $\pgfdom$ have a least upper bound $\sup \chain \in \pgfdom$.
The maxima in $(\pgfdom, \sqsubseteq)$ are precisely the PGF which are not a sub-PGF.


\subsection{From Programs to PGF Transformers}
\label{sec:explainTransformer}

Next we explain how distribution transformation works using (P)GF (cf.\ \Cref{fig:cheatSheet}). 
This is in contrast to the PGF semantics from~\cite{lopstr/klinkenberg20} which operates  on infinite sums in a non-constructive fashion.

\begin{definition}[The PGF Transformer $\sem{P}$]
    Let $P$ be a $\pplname$-program.
    The \emph{PGF transformer} $\sem{P} \colon \pgfdom \to \pgfdom$ is defined inductively on the structure of $P$ through the second column in \Cref{fig:synsem}. 
\end{definition}

We show in \Cref{thm:semPprops} below that $\sem{P}$ is well-defined.
For now, we go over the statements in the language $\pplname$ and explain the semantics.

\paragraph{Sequential Composition.}
The semantics of $\compose{P_1}{P_2}$ is straightforward and intuitive:
First execute $P_1$ on $g$ and then $P_2$ on $\sem{P_1}(g)$, i.e., $\sem{\compose{P_1}{P_2}}(g) = \sem{P_2}(\sem{P_1}(g))$.
The fact that our semantics transformer moves \emph{forwards} through the program -- as program interpreters usually do -- is due to this definition.

\paragraph{Conditional Branching.}
To translate $\ITE{\progvar{x} < n}{P_1}{P_2}$, we follow the standard procedure which partitions the input distribution according to $\progvar{x} < n$ and $\progvar{x} \geq n$, processes the two parts independently and finally recombines the results~\cite{DBLP:journals/jcss/Kozen85}.
We realize the partitioning using the (formal) \emph{Taylor series expansion}.
This is feasible because we only allow \emph{rectangular} guards of the form $\progvar{x} < n$, where $n$ is a constant.
Thus, for a given input PGF $g$, the \emph{filtered PGF} $g_{\progvar{x} < n}$ is obtained through expanding $g$ in its first $n$ terms.
The $\ELSE$-part is obviously $g_{\progvar{x} \geq n} = g - g_{\progvar{x} < n}$.
We then evaluate $\sem{P_1}(g_{\progvar{x} < n}) + \sem{P_2}(g_{\progvar{x} \geq n})$ recursively.

\paragraph{Assigning a Constant.}
Technically, our semantics realizes an assignment $\asgn{\progvar{x}}{n}$ in two steps:
It first sets $\progvar{x}$ to $0$ and then increments it by $n$.
The former is achieved by substituting $X$ for $1$ which corresponds to computing the marginal distribution in all variables except $X$.
For example,
    \begin{alignat*}{3}
        &\annotate{0.5 X Y^2 + 0.5 X^2 Y^3} & \qquad \qquad& \annotate{g} \\
        &\asgn{\progvar{x}}{5} & \qquad \qquad& P \\
        & \annotate{(0.5 Y^2 + 0.5 Y^3)X^5} & \qquad \qquad& \annotate{\sem{P}(g)} \\
        & \annotate{0.5 X^5Y^2  + 0.5 X^5 Y^3 } & \qquad \qquad& \annotate{\langle\, \textit{reform.\ of prev.\ line} \,\rangle}
    \end{alignat*}
where the rightmost four lines explain this annotation style~\cite{lopstr/klinkenberg20}.
Note that $0.5 Y^2 + 0.5 Y^3$ is indeed the marginal of the input distribution in $Y$.

\paragraph{Decrementing a Variable.}
Since our program variables cannot take negative values, we define $\decr{\progvar{x}}$ as $\max(\progvar{x} {-} 1, 0)$, i.e., $\progvar{x}$ \emph{monus} (modified minus) $1$.
Technically, we realize this through $\ITE{\progvar{x}<1}{\pskip}{\decr{\progvar{x}}}$, i.e., we apply the decrement only to the portion of the input distribution where $\progvar{x} \geq 1$.
The decrement itself can then be carried out through \enquote{multiplication by $X^{-1}$}.
Note that $X^{-1}$ is not an element of $\fpsring{X}$ because $X$ has no inverse.
Instead, the operation $gX^{-1}$ is an alias for $\leftshift{g}$ which shifts $g$ \enquote{to the left} in dimension $X$.
To implement the semantics on top of existing computer algebra software, it is very handy to perform the multiplication by $X^{-1}$ instead.
This is justified because for PGF $g$ with $\subsFPSVarFor{g}{X}{0} = 0$, $\leftshift{g}$ and $gX^{-1}$ are equal.

\paragraph{The $\mathtt{iid}$ Statement.}
The semantics of $\incrasgn{\progvar{x}}{\iid{D}{\progvar{y}}}$ relies on the fact that 
\begin{align}
    \label{eq:iidsum}
    T_1 \sim \sem{D} ~ \ldots ~ T_n \sim \sem{D} \qqimplies \sum\nolimits_{i=1}^n T_i \sim \sem{D}^n ~,
\end{align}%
where $X \sim g$ means that r.v.\ $X$ is distributed according to PGF $g$ (see, e.g.,~\cite[p.~450]{tijms2003first}).
The $\mathtt{iid}$ statement generalizes this observation further:
If $n$ is not a constant but a random (program) variable $\progvar{y}$ with PGF $h(Y)$, then we perform the \emph{substitution} $\subsFPSVarFor{h}{Y}{\sem{D}}$ (i.e., replace $Y$ by $\sem{D}$ in $h$) to obtain the PGF of the sum of $\progvar{y}$-many i.i.d.\ samples from $D$.
We slightly modify this substitution to $\subsFPSVarFor{g}{Y}{Y\subsFPSVarFor{\sem{D}}{T}{X}}$ in order to (i) not alter $\progvar{y}$, and (ii) account for the increment to $\progvar{x}$.
For example,
\begin{align*}
    &\annotate{0.2 + 0.3Y + 0.5Y^2}\\
    &\incrasgn{\progvar{x}}{\iid{\bernoulli{0.5}}{\progvar{y}}} \\
    &\annotate{0.2 + 0.3 Y(0.5 + 0.5X) + 0.5Y^2(0.5 + 0.5X)^2} \\
    &\annotate{0.2 + 0.15 Y + 0.125 Y^2 + 0.15 X Y + 0.25 X Y^2 + 0.125 X^2 Y^2 } ~.
\end{align*}

\paragraph{The $\mathtt{while}$-Loop.}
The fixed point semantics of the while loop is standard~\cite{lopstr/klinkenberg20,DBLP:journals/jcss/Kozen85} and reflects the intuitive \emph{unrolling rule}, namely that $\WHILEDO{\guard}{P}$ is equivalent to $\ITE{\guard}{\compose{P}{\WHILEDO{\guard}{P}}}{\pskip}$.
Indeed, the fixed point formula in \Cref{fig:synsem} can be derived using the semantics of $\mathtt{if}$ discussed above.
We revisit this fixed point characterization in \Cref{sec:fpinduction}.

\paragraph{Properties of $\sem{P}$.}
Our PGF semantics has the property that all programs -- except while loops -- are able to operate on the input PGF in (rational) \emph{closed form}, i.e., they never have to expand the input as an infinite series (which is of course impossible in practice).
More formally:

\begin{theorem}[Closed-Form Preservation]
    \label{thm:closedform}
    Let $P$ be a \emph{loop-free} $\pplname$ program, and let $g = h/f \in \pgfdom$ be in rational closed form.
    Then we can compute a rational closed form of $\sem{P}(g)  \in \pgfdom$ by applying the transformations in \Cref{fig:synsem}.
\end{theorem}

The proof is by induction over the structure of $P$ noticing that all the necessary operations (substitution, differentiation, etc.) preserve rational closed forms, see~\ifcameraready\cite[Appx.~D]{DBLP:journals/corr/abs-2205-01449}\else\Cref{sec:fpsbackground}\fi.
A slight extension of our syntax, e.g., admitting non-rectangular guards, renders that closed forms are not preserved, see~\ifcameraready\cite[Appx.~B]{DBLP:journals/corr/abs-2205-01449}\else\Cref{app:nonrect}\fi.
Moreover, $\sem{P}$ has the following \emph{healthiness}~\cite{McIver04} properties:

\begin{restatable}[Properties of $\sem{P}$]{theorem}{restateSemPProps}
    \label{thm:semPprops}
    The PGF transformer $\sem{P}$ is
    \begin{itemize}
        \item a well-defined function $\pgfdom \to \pgfdom$~,
        \item \emph{continuous}, i.e., $\sem{P}(\sup \chain) = \sup \sem{P}(\chain)$ for all chains $\chain \subseteq \pgfdom$~,
        \item \emph{linear}, i.e., $\sem{P}(\fsum_{\sigma \in \N^k} g(\sigma) \vec{X}^\sigma) = \fsum_{\sigma \in \N^k} g(\sigma) \sem{P}(\vec{X}^\sigma)$ for all $g \in \pgfdom$~.
    \end{itemize}
\end{restatable}

\subsection{Probabilistic Termination}

Due to the presence of possibly unbounded $\while$-loops, a $\pplname$-program does not necessarily halt, or may do so only with a certain probability.
Our semantics naturally captures the termination probability.

\begin{definition}[AST]
        \label{def:ast}
    A $\pplname$-program $P$ is called \emph{almost-surely terminating} (AST) for PGF $g$ if $\subsFPSVarFor{\sem{P}(g)}{\vec{X}}{\vec{1}} = \subsFPSVarFor{g}{\vec{X}}{\vec{1}}$, i.e., if it does not leak probability mass.
    $P$ is called \emph{universally} AST (UAST) if it is AST for all $g \in \pgfdom$.
\end{definition}

Note that all loop-free $\pplname$-programs are UAST.
In this paper, (U)AST only plays a minor role. 
Nonetheless, the proof rule below yields a stronger result (cf. \Cref{thm:fpToEquiv}) if the program is UAST.
There exist various of techniques and tools for proving (U)AST~\cite{DBLP:conf/fm/MoosbruggerBKK21,DBLP:journals/pacmpl/McIverMKK18,chatterjeeFOPP20}.

\section{Reasoning about Loops}
\label{sec:loops}

We now focus on loopy programs $L = \WHILEDO{\guard}{P}$.
Recall from \Cref{fig:synsem} that $\sem{L} \colon \pgfdom \to \pgfdom$ is defined as the \emph{least fixed point} of a higher order functional
\[
\charfun{\guard}{P} \colon (\pgfdom \to \pgfdom) \tto (\pgfdom \to \pgfdom) ~.
\]
Following~\cite{lopstr/klinkenberg20}, we show that $\charfun{\guard}{P}$ is sufficiently well-behaved to allow reasoning about loops by \emph{fixed point induction}.

\subsection{Fixed Point Induction}
\label{sec:fpinduction}

To apply fixed point induction, we need to lift our domain $\pgfdom$ from \Cref{sec:pgfdom} by one order to $(\pgfdom \to \pgfdom)$, the domain of \emph{PGF transformers}.
This is because the functional $\charfun{\guard}{P}$ operates on PGF transformers and can thus be seen as a second-order function (this point of view regards PGF as first-order objects).
Recall that in contrast to this, the function $\sem{P}$ is first-order -- it is just a PGF transformer.
The order on $(\pgfdom \to \pgfdom)$ is obtained by lifting the order $\sqsubseteq$ on $\pgfdom$ pointwise (we denote it with the same symbol $\sqsubseteq$).
This implies that $(\pgfdom \to \pgfdom)$ is also an $\omega$-complete partial order.
We can then show that $\charfun{\guard}{P}$ (see \Cref{fig:synsem}) is a continuous function.
With these properties, we obtain the following induction rule for upper bounds on $\sem{L}$, cf.~\cite[Theorem 6]{lopstr/klinkenberg20}:

\begin{restatable}[{Fixed Point Induction for Loops}]{lemma}{restateFpInduction}
    \label{thm:fpInduction}
    Let $L = \WHILEDO{\guard}{P}$ be a $\pplname$-loop.
    Further, let $\psi \colon \pgfdom \to \pgfdom$ be a PGF transformer.
    Then
    \begin{align*}
        \charfun{\guard}{P}(\psi) \ssqsubseteq \psi
        \qqimplies
        \sem{L} \ssqsubseteq \psi
        ~.
    \end{align*}
\end{restatable}

The goal of the rest of the paper is to \emph{apply the rule from \Cref{thm:fpInduction} in practice}.
To this end, we must somehow specify an \emph{invariant} such as $\psi$ by finite means.
Since $\psi$ is of type $(\pgfdom \to \pgfdom)$, we consider $\psi$ as a program $I$ -- more specifically, a $\pplname$-program -- and identify $\psi = \sem{I}$.
Further, by definition
\[
    \charfun{\guard}{P}(\sem{I})
    \eeq
    \sem{\ITE{\guard}{P \fatsemi I}{\pskip}} ~,
\]
and thus the term $\charfun{\guard}{P}(\sem{I})$ is also a PGF-transformer expressible as a $\pplname$-program.
These observations and \Cref{thm:fpInduction} imply the following:
\begin{lemma}
    \label{thm:fpToEquiv}
    Let $L = \WHILEDO{\guard}{P}$ and $I$ be $\pplname$-programs.
    Then
    \begin{align}
        \sem{\ITE{\guard}{P \fatsemi I}{\pskip}} \ssqsubseteq \sem{I} \qimplies \sem{L} \ssqsubseteq \sem{I} ~.
    \end{align}
    Further, if $L$ is UAST (\Cref{def:ast}), then 
    \begin{align}
        \label{eq:equiviff}
        \sem{\ITE{\guard}{P \fatsemi I}{\pskip}} \eeq \sem{I} \qiff \sem{L} \eeq \sem{I} ~.
    \end{align}
\end{lemma}

\Cref{thm:fpToEquiv} effectively reduces checking whether $\psi$ given as a $\pplname$-program $I$ is an invariant of $L$ to checking \emph{equivalence} of $\ITE{\guard}{P \fatsemi I}{\pskip}$ and $I$ provided $L$ is UAST.
If $I$ is loop-free, then the latter two programs are both loop-free and we are left with the task of proving whether they yield the same output distribution for all inputs.
We now present a solution to this problem.

\subsection{Deciding Equivalence of Loop-free Programs}
\label{sec:equivalence}

Even in the absence of loops, deciding if two given $\pplname$-programs are equivalent is non-trivial as it requires reasoning about infinitely many -- possibly infinite-support -- distributions on program variables.
In this section, we first show that $\sem{P_1} = \sem{P_2}$ is \emph{decidable} for loop-free $\pplname$ programs $P_1$ and $P_2$, and then use this result together with \Cref{thm:fpToEquiv} to obtain the main result of this paper. 

\subsubsection{SOP: Second-Order PGF.}

Our goal is to check if $\sem{P_1}(g) = \sem{P_2}(g)$ for \emph{all} $g \in \pgfdom$.
To tackle this, we encode whole \emph{sets} of PGF into a single object -- an FPS we call \emph{second-order PGF} (SOP).
To define SOP, we need a slightly more flexible view on FPS.
Recall from \Cref{def:fps} that a $k$-dim.\ FPS is an array $f \colon \N^k \to \R$.
Such an $f$ can be viewed equivalently as an $l$-dim.\ array with $(k{-}l)$-dim.\ arrays as entries.
In the formal sum notation, this is reflected by partitioning $\vec{X} = (\vec{Y}, \vec{Z})$ and viewing $f$ as an FPS in $\vec{Y}$ \emph{with coefficients that are FPS in the other indeterminates $\vec{Z}$}.
For example,
\begin{align*}
    (1-Y)^{-1}(1-Z)^{-1} &\eeq 1 + Y + Z + Y^2 + YZ + Z^2 + \ldots\\
    &\eeq (1-Z)^{-1} + (1-Z)^{-1}Y + (1-Z)^{-1}Y^2 + \ldots
\end{align*}
where in the lower line the coefficients $(1{-}Z)^{-1}$ are considered elements in $\fpsring{Z}$.

\begin{definition}[SOP]
    Let $\vec{U}$ and $\vec{X}$ be disjoint sets of indeterminates.
    A formal power series $f \in \R[[\vec{U},\vec{X}]]$ is a \emph{second-order PGF (SOP)} if
    \[
        f = \sum\nolimits_{\tau \in \N^{\abs{\vec{U}}}} f(\tau) \vec{U}^\tau \quad (\text{with } f(\tau) \in \fpsring{\vec{X}}) \qqimplies \forall \tau \colon f(\tau) \in \pgfdom ~.
    \]
\end{definition}
That is, an SOP is simply an FPS whose coefficients are PGF -- instead of generating a sequence of probabilities as PGF do, it generates a \emph{sequence of distributions}.
An (important) example SOP is 
\begin{align}
    \label{eq:sopdirac}
    f_{\mathit{dirac}}
    \eeq
    (1-XU)^{-1}
    \eeq
    1 + XU + X^2U^2 + \ldots ~\in \fpsring{U,X} ~,
\end{align}
i.e., for all $i\geq 0$, $f_{\mathit{dirac}}(i) = X^i = \sem{\dirac{i}}$.
As a second example consider $f_{\mathit{binom}} = \subsFPSVarFor{f_{\mathit{dirac}}}{X}{0.5 + 0.5X}$; it is clear that $f_{\mathit{binom}}(i) = (0.5 + 0.5X)^i = \sem{\binomial{0.5}{i}}$ for all $i \geq 0$.
Note that if $\vec{U} = \emptyset$, then SOP and PGF coincide.
For fixed $\vec{X}$ and $\vec{U}$, we denote the set of all second-order PGF with $\sopdom$.

\subsubsection{SOP Semantics of $\pplname$.}

The appeal of SOP is that, syntactically, they are still formal power series, and some can be represented in closed form just like PGF.
Moreover, we can readily extend our PGF transformer $\sem{P}$ to an SOP transformer $\sem{P} \colon \sopdom \to \sopdom$.
A key insight of this paper is that -- without any changes to the rules in \Cref{fig:synsem} -- applying $\sem{P}$ to an SOP is the same as applying $\sem{P}$ \emph{simultaneously} to all the PGF it subsumes:

\begin{restatable}{theorem}{restateSoptrans}
    \label{thm:soptrans}
    Let $P$ be a $\pplname$-program.
    The transformer $\sem{P} \colon \sopdom \to \sopdom$ is well-defined.
    Further, if $f = \sum_{\tau \in \N^{\abs{\vec{U}}}} f(\tau) \vec{U}^\tau$ is an SOP, then 
    \[
        \sem{P}(f) \eeq \sum\nolimits_{\tau \in \N^{\abs{\vec{U}}}} \sem{P}(f(\tau)) \vec{U}^\tau ~.
    \]
\end{restatable}

\subsubsection{An SOP Transformation for Proving Equivalence.}

We now show how to exploit \Cref{thm:soptrans} for equivalence checking.
Let $P_1$ and $P_2$ be (loop-free) $\pplname$-programs;
we are interested in proving whether $\sem{P_1} = \sem{P_2}$.
By linearity it holds that $\sem{P_1} = \sem{P_2}$ iff $\sem{P_1}(\vec{X}^\sigma) = \sem{P_2}(\vec{X}^\sigma)$ for all $\sigma \in \N^k$, i.e., to check equivalence it suffices to consider all (infinitely many) point-mass PGF as inputs.

\begin{restatable}[SOP-Characterisation of Equivalence]{lemma}{restateEquivCharac}
    \label{thm:equivCharac}
    Let $P_1$ and $P_2$ be $\pplname$-programs with $\vars(P_i) \subseteq \{\progvar{x_1},\ldots,\progvar{x_k}\}$ for $i \in \{1,2\}$.
    Further, consider a vector $\vec{U} = (U_1,\ldots,U_k)$ of meta indeterminates, and let $g_{\vec{X}}$ be the SOP
    \[
        g_{\vec{X}} \eeq (1 - X_1 U_1)^{-1} (1 - X_2 U_2)^{-1} \cdots (1 - X_k U_k)^{-1} \iin \fpsring{\vec{U},\vec{X}} ~.
    \]
    Then $\sem{P_1} = \sem{P_2}$ if and only if $\sem{P_1}(g_{\vec{X}}) = \sem{P_2}(g_{\vec{X}})$.
\end{restatable}
The proof of \Cref{thm:equivCharac} (\ifcameraready%
see~\cite[Appx.~F.5]{DBLP:journals/corr/abs-2205-01449}\else%
given in \Cref{proof:equivCharac}\fi) relies on \Cref{thm:soptrans} and the fact that the \emph{rational} SOP $g_{\vec{X}}$ generates all (multivariate) point-mass PGF; in fact it holds that
$g_{\vec{X}} = \sum_{\sigma \in \N^k} \vec{X}^\sigma \vec{U}^\sigma$, i.e., $g_{\vec{X}}$ generalizes $f_{\mathit{dirac}}$ from \eqref{eq:sopdirac}.
It follows:

\begin{restatable}{lemma}{restateLoopFreeDecidability}
    \label{thm:loop-free-decidability}
    $\sem{P_1} = \sem{P_2}$ is decidable for loop-free  $\pplname$-programs $P_1, P_2$.
\end{restatable}

Our main theorem follows immediately from \Cref{thm:fpToEquiv} and \Cref{thm:loop-free-decidability}:

\begin{theorem}
    \label{cor:loop-equiv-decidability}
    Let $L = \WHILEDO{\guard}{P}$ be UAST with loop-free body $P$ and $I$ be a loop-free $\pplname$-program.
    It is decidable whether $\sem{L} = \sem{I}$.
\end{theorem}

\begin{example}
    In \Cref{fig:equivCheckEx} we prove that the two UAST programs $L$ and $I$\\
    \begin{minipage}{\textwidth}
        \centering
        \begin{minipage}[t]{0.3\textwidth}
            \begin{align*}
            & \WHILE{\progvar{n} > 0} \\
            & \qquad \PCHOICE{\asgn{\progvar{n}}{\progvar{n} - 1}}{\nicefrac{1}{2}}{\asgn{\progvar{c}}{\progvar{c} + 1}} ~\} \\
            \end{align*}
        \end{minipage}
        ~~~~
        \begin{minipage}[t]{0.3\textwidth}
            \begin{align*}
            & \incrasgn{\progvar{c}}{\iid{\geometric{\nicefrac{1}{2}}}{\progvar{n}}} \,\fatsemi \\
            & \asgn{\progvar{n}}{0} \\
            \end{align*}
        \end{minipage}
    \end{minipage}\\
    are equivalent (i.e., $\sem{L} = \sem{I}$) by showing that $\sem{\IF{\progvar{n} > 0} \{\compose{P}{I}\}} = \sem{I}$ as suggested by \Cref{thm:fpToEquiv}.
    The latter is achieved as in \Cref{thm:equivCharac}: 
    We run both programs on the input SOP $g_{N,C} = (1 - NU)^{-1} (1 - CV)^{-1}$, where $U, V$ are meta indeterminates corresponding to $N$ and $C$, respectively, and check if the results are equal.
    Note that $I$ is the loop-free specification from \Cref{exmp:russian-roulette}; thus by transitivity, the loop $L$ is equivalent to the loop in \Cref{fig:russianroulette}.
    \qedTT
\end{example}


\begin{figure}[t]
    \centering
    \begin{adjustbox}{max width=\textwidth,max height=35mm}
        \begin{minipage}{0.69\textwidth}
            \begin{align*}
            & \annotate{(1-NU)^{-1}(1-CV)^{-1} \textcolor{gray}{= g_{N,C} =: g_0}} \\
            & \IF{\progvar{n} > 0} \{ \\
            & \qquad \annotate{ (1-CV)^{-1} ((1-NU)^{-1} -1 )  \textcolor{gray}{= g_0 - \subsFPSVarFor{g_0}{N}{0} =: g_1}} \\
            & \qquad  \{\, \asgn{\progvar{n}}{\progvar{n} - 1} \\
            & \qquad \qquad  \annotate{N^{-1}(1-CV)^{-1} ((1-NU)^{-1} -1 ) \textcolor{gray}{ = g_1N^{-1} =: g_2} } \\
            & \qquad  \} \, [0.5] \, \{ \, \incrasgn{\progvar{c}}{1} ~\} \,\fatsemi\\
            & \qquad \qquad \annotate{ C(1-CV)^{-1} ((1-NU)^{-1} -1 ) \textcolor{gray}{ = g_1C =: g_3} } \\
            & \qquad \annotate{(2N(1-CV))^{-1} + C(2(1-CV))^{-1}) ((1-NU)^{-1} -1 ) \textcolor{gray}{ = 0.5g_2 + 0.5g_3 =: g_4} } \\
            & \qquad \incrasgn{\progvar{c}}{\iid{\geometric{\nicefrac{1}{2}}}{\progvar{n}}} \,\fatsemi \\
            & \qquad \annotate{(2-C)(2N(1-CV))^{-1} + C(2(1-CV))^{-1}) ((2-C)(2-C-NU)^{-1} -1 )} \\
            & \qquad \qquad \textcolor{gray}{= \subsFPSVarFor{g_4}{N}{N(2-C)^{-1}}=: g_5} \\
            & \qquad \asgn{\progvar{n}}{0}  ~\}  \\
            & \qquad \annotate{(1-CV)^{-1} ( (2-C) (2-C-U)^{-1} - 1 ) \textcolor{gray}{ = \subsFPSVarFor{g_5}{N}{1} = g_6} }\\
            & \annotate{\textcolor{RoyalPurple}{(1-CV)^{-1}  (2-C) (2-C-U)^{-1}} \textcolor{gray}{ = g_6 + \subsFPSVarFor{g_0}{N}{0}}}
            \end{align*}
        \end{minipage}
        \qquad
	    \begin{minipage}{0.3\textwidth}
	    	\begin{align*}
	    	& \annotate{(1-NU)^{-1}(1-CV)^{-1}} \\
	    	& \incrasgn{\progvar{c}}{\iid{\geometric{\nicefrac{1}{2}}}{\progvar{n}}} \,\fatsemi \\
	    	& \annotate{(1-CV)^{-1}  (2-C) (2-C-NU)^{-1}} \\
	    	& \asgn{\progvar{n}}{0} \\
	    	& \annotate{\textcolor{RoyalPurple}{(1-CV)^{-1}  (2-C) (2-C-U)^{-1}}}
	    	\end{align*}
	    \end{minipage}
    \end{adjustbox}
    \caption{
        Program equivalence follows from the equality of the resulting SOP (\Cref{thm:equivCharac}).
    }
    \label{fig:equivCheckEx}
\end{figure}
\section{Case Studies}
\label{sec:case-studies}
We have implemented our techniques in Python as a prototype called \toolname\footnote{\ifblindreview Link omitted due to blind reviews\else\faGithub~\url{https://github.com/LKlinke/Prodigy}\fi.}: PRObability DIstributions via GeneratingfunctionologY.
By interfacing with different computer algebra systems (CAS), e.g., \texttt{Sympy}~\cite{SymPy} and \texttt{GiNaC}~\cite{DBLP:journals/jsc/BauerFK02,ginac} -- as backends for symbolic computation of PGF and SOP semantics -- {\toolname} decides whether a given probabilistic loop agrees with an (invariant) specification encoded as a loop-free $\pplname$ program.
Furthermore, it supports efficient queries on various quantities associated with the output distribution.

In what follows, we demonstrate in particular the applicability of our techniques to programs featuring stochastic dependency, parametrization, and nested loops.
The examples are all presented in the same way:
the iterative program on the left side and its corresponding specification on the right. The presented programs are all UAST, given the parameters are instantiated from a suitable value domain.\footnote{Parameters of \Cref{exmp:dueling_cowboys} have to be instantiated with a probability value in $(0,1)$.}
For each example, we report the time 
for performing the equivalence check on a 2,4GHz Intel i5 Quad-Core processor with 16GB RAM running macOS Monterey 12.0.1.
Additional examples can be found in~\ifcameraready\cite[Appx.~E]{DBLP:journals/corr/abs-2205-01449}\else\Cref{apx:casestudies}\fi.

\begin{figure}[t]
	\centering
	\begin{adjustbox}{max width=\textwidth, max height=13mm}
		\begin{minipage}{\textwidth}
			\begin{minipage}[t]{0.5\textwidth}
				\begin{align*}
					& \WHILE{\progvar{c} > 0}\\
					& \qquad \PCHOICE{\ASSIGN{\progvar{n}}{\progvar{n} + 1}}{\nicefrac{1}{2}}{\ASSIGN{\progvar{m}}{\progvar{m} + 1}}\fatsemi\\
					& \qquad \ASSIGN{\progvar{c}}{\progvar{c} -1}\fatsemi\\
					& \qquad \ASSIGN{\progvar{tmp}}{0}\\
					&\}
				\end{align*}
			\end{minipage}
			\qquad
			\begin{minipage}[t]{0.4\textwidth}
				\begin{align*}
					& \IF{\progvar{c} > 0}\{\\
					& \qquad \ASSIGN{\progvar{tmp}}{\binomial{\nicefrac{1}{2}}{\progvar{c}}}\fatsemi\\
					& \qquad \incrasgn{ \progvar{m} }{\progvar{tmp} }\fatsemi ~
					\incrasgn{ \progvar{n} }{\progvar{c}  -  \progvar{tmp}}\fatsemi\\
					& \qquad \ASSIGN{\progvar{c}}{0}\,\fatsemi\\
					& \qquad \ASSIGN{\progvar{tmp}}{0} \,\}
				\end{align*}
			\end{minipage}
		\end{minipage}
	\end{adjustbox}
	\caption{Generating complementary binomial distributions (for $\progvar{n}, \progvar{m}$) by coin flips. $\binomial{\nicefrac{1}{2}}{\progvar{c}}$ is an alias for $\iid{\bernoulli{\nicefrac{1}{2}}}{\progvar{c}}$.}
	\label{fig:depbern}
\end{figure}

\begin{example}[Complementary Binomial Distributions]
\label{exmp:dep_bern}
We show that the program in \Cref{fig:depbern} generates  a joint distribution on  $\progvar{n}, \progvar{m}$ such that both $\progvar{n}$ and $\progvar{m}$ are binomially distributed with support $\progvar{c}$ and are complementary in the sense that $\progvar{n} + \progvar{m} = \progvar{c}$ holds certainly (if $\progvar{n}=\progvar{m=0}$ initially, otherwise the variables are incremented by the corresponding amounts).
{\toolname} automatically checks that the loop agrees with the specification in 18.3ms.
The resulting distribution can then be analyzed for any given input PGF $g$ by computing $\sem{I}(g)$, where $I$ is the loop-free program.
For example, for input $g = C^{10}$, the distribution as computed by {\toolname} has the \emph{factorized} closed form $(\frac{M+N}{2})^{10}$.
The CAS backends exploit such factorized forms to perform algebraic manipulations more efficiently compared to fully expanded forms.
For instance, we can evaluate the queries $\expected{m^3+ 2mn + n^2} = 235$, or \mbox{$Pr(m >7 \land n < 3) = 7/128$}, almost instantly.
%
\qedTT
\end{example}

\begin{figure}[t]
	\begin{center}
		\begin{adjustbox}{max width=\textwidth, max height=16mm}
			\begin{minipage}{\textwidth}
				\begin{minipage}[t]{0.4\textwidth}
					\begin{align*}
						& \WHILE{\progvar{c} = 1 ~\wedge~ \progvar{t} \leq 1
						}\\
						& \qquad \IF{\progvar{t} = 0} \{\\
						& \qquad \qquad \PCHOICE{ \ASSIGN{\progvar{c}}{0} } {a}{\ASSIGN{\progvar{t}}{1}}\\
						& \qquad \} ~\pelse~ \{\\
						& \qquad \qquad \PCHOICE{ \ASSIGN{\progvar{c}}{0} } {b}{\ASSIGN{\progvar{t}}{0}}\\
						&\qquad \} \\
						& \}
					\end{align*}
				\end{minipage}
				\qquad~~~~
				\begin{minipage}[t]{0.5\textwidth}
					\begin{align*}
						& \IF{\progvar{c} = 1 ~\wedge~ \progvar{t} \leq 1
						}\{\\
						& \qquad \ASSIGN{\progvar{c}}{0}\\
						& \qquad \IF{\progvar{t} = 0}\{\\
						& \qquad \qquad \ASSIGN{\progvar{t}}{\bernoulli{\nicefrac{(1-a)b}{a + b -ab}}}\fatsemi\\
						& \qquad \} ~\ELSE \{\\
						& \qquad \qquad \ASSIGN{\progvar{t}}{\bernoulli{ \nicefrac{b}{a+b-ab}}}\fatsemi\\
						& \qquad \} \,\}
					\end{align*}
				\end{minipage}
			\end{minipage}
		\end{adjustbox}
		\caption{A program modeling two dueling cowboys with parametric hit probabilities.}
		\label{fig:dueling_cowboys}
	\end{center}
\end{figure}

\begin{example}[Dueling Cowboys~\textnormal{\cite{McIver04}}]
\label{exmp:dueling_cowboys}
The program in \Cref{fig:dueling_cowboys} 
models a duel of two cowboys with \emph{parametric} hit probabilities $\progvar{a}$ and $\progvar{b}$. 
Variable $\progvar{t}$ indicates the cowboy who is currently taking his shot, and $\progvar{c}$ monitors the state of the duel ($\progvar{c} = 1$: duel is still running, $\progvar{c} = 0$: duel is over).
$\toolname$ automatically verifies the specification 
in 11.97ms.
We defer related problems -- e.g., \emph{synthesizing} parameter values to meet a parameter-free specification -- to future work.
\qedTT
\end{example}

\begin{figure}[t]
	\begin{adjustbox}{max width=\textwidth}
		\begin{minipage}[t]{0.3\textwidth}
			\begin{align*}
				& \WHILE{\progvar{x} > 0} \\
				& \qquad \ASSIGN{\progvar{y}}{1}\fatsemi\\
				& \qquad \WHILE{\progvar{y} = 1}\\
				& \qquad \qquad \PCHOICE{\ASSIGN{\progvar{y}}{0}}{\nicefrac{1}{2}}{\ASSIGN{\progvar{x}}{\progvar{x} + 1}} \,\} \fatsemi\\
				& \qquad \ASSIGN{\progvar{x}}{\progvar{x} - 1}\fatsemi\\
				& \qquad \incrasgn{\progvar{c}}{1}\fatsemi\\
				& \} 
			\end{align*}
		\end{minipage}
		\qquad
		\begin{minipage}[t]{0.3\textwidth}
			\begin{align*}
				& \comment{inner invariant}\\
				& \IF{\progvar{y} = 1} \{\\
				& \qquad \incrasgn{\progvar{x}}{\geometric{\nicefrac{1}{2}}}\fatsemi\\
				& \qquad \ASSIGN{\progvar{y}}{0}\fatsemi\\
				& \}
			\end{align*}
		\end{minipage}
		\qquad
		\begin{minipage}[t]{0.3\textwidth}
			\begin{align*}
				& \comment{outer invariant}\\
				& \IF{\progvar{x} > 0}\{\\
				%
				%
				& \qquad \ASSIGN{\progvar{c}}{\iid{\catalan{\nicefrac{1}{2}}}{\progvar{x}}}\fatsemi\\
				& \qquad \ASSIGN{\progvar{x}}{0}\fatsemi\\
				& \qquad \ASSIGN{\progvar{y}}{0}\\
				&\}
			\end{align*}
		\end{minipage}
	\end{adjustbox}
	\caption{Nested loops with invariants for the inner and outer loop.}
	\label{fig:nested_loops}
\end{figure}

\begin{example}[Nested Loops]
\label{exmp:nested_loops}
The inner loop of the program in \Cref{fig:nested_loops} modifies $\progvar{x}$ which influences the termination behavior of the outer loop.
Intuitively, the program models a random walk on $\N$:
In every step, the value of the current position $\progvar{x}$ changes by some random $\delta \in \{-1,0,1,2,\ldots\}$ such that $\delta +1 $ is geometrically distributed.
The example demonstrates how our technique enables \emph{compositional} reasoning.
We first provide a loop-free specification for the inner loop, prove its correctness, and then simply \emph{replace} the inner loop by its specification, yielding a program without nested loops.
This feature is a key benefit of reusing the loop-free fragment of $\pplname$ as a specification language.
Moreover, existing techniques that cannot handle nested loops can profit from it; in fact, we can prove the overall program to be UAST using the rule of~\cite{DBLP:journals/pacmpl/McIverMKK18}.
Interestingly, the outer loop has \emph{infinite expected runtime} (for any input distribution where the probability that $\progvar{x} > 0$ is positive). We can prove this by \emph{querying the expected value} of the program variable $\progvar{c}$ in the resulting output distribution.
The automatically computed result is $\infty$, which indeed proves that the expected runtime of this program is not finite.
This example furthermore shows that our technique can be generalized beyond rational functions since the PGF of the $\catalan{p}$ distribution is $(1 - \sqrt{1 - 4 p (1{-}p) T} ) \,/\, 2p$, i.e., algebraic but not rational.
We leave a formal generalization of the decidability result from \Cref{cor:loop-equiv-decidability} to algebraic functions for future work.
{\toolname} verifies this example in 29.17ms.
\qedTT
\end{example}

\paragraph*{Scalability Issue.}
It is not difficult to construct programs where {\toolname} poorly scales: its performance depends highly on the number of consecutive probabilistic branches and the size of the constant $n$ in guards (requiring $n$-th order PGF derivation, cf.\ \Cref{fig:synsem}).
\section{Related Work}
\label{sec:related-work}

This section surveys research efforts 
that are highly related to our approach in terms of semantics, inference, and equivalence checking of probabilistic programs.


\paragraph*{Forward Semantics of Probabilistic Programs.}
Kozen established in his seminal work~\cite{DBLP:journals/jcss/Kozen81} a generic way of giving forward, denotational semantics to probabilistic programs as \emph{distribution transformers}. Klinkenberg et al.~\cite{lopstr/klinkenberg20} instantiated Kozen's semantics 
as PGF transformers. We refine the PGF semantics substantially such that it enjoys the following crucial properties:
\begin{enumerate*}[label=(\roman*)]
\item our PGF transformers (when restricted to loop-free $\pplname$ programs) preserve closed-form PGF and thus are effectively constructable. In contrast, the existing PGF semantics in~\cite{lopstr/klinkenberg20} operates on infinite sums in a non-constructive fashion;
\item our PGF semantics naturally extends to SOP, which serves as the key to reason about the exact behavior of unbounded loops (under possibly uncountably many inputs) in a fully automatic manner. The PGF semantics in~\cite{lopstr/klinkenberg20}, however, supports only (over-)approximations of looping behaviors and can hardly be automated; and
\item our PGF semantics is capable of interpreting program constructs like i.i.d.\ sampling that is of particular interest in practice.
\end{enumerate*}

\paragraph*{Backward Semantics of Probabilistic Programs.}
Many verification systems for probabilistic programs make use of backward, denotational semantics -- most pertinently, the \emph{weakest preexpectation} (WP) calculi~\cite{McIver04,DBLP:phd/dnb/Kaminski19} as a quantitative extension of Dijkstra's weakest preconditions~\cite{DBLP:journals/cacm/Dijkstra75}. The WP of a probabilistic program $C$ w.r.t.\ a postexpectation $g$, denoted by $\wpsymbol \llbracket C \rrbracket(g)(\cdot)$, maps every initial program state $\sigma$ to the expected value of $g$ evaluated in final states reached after executing $C$ on $\sigma$. In contrast to Dijkstra's predicate transformer semantics which admits also strongest postconditions, the counterpart of \enquote{strongest postexpectations} does unfortunately not exist~\cite[Chap.\ 7]{DBLP:phd/ethos/Jones90}, thereby not amenable to forward reasoning. We remark, in particular, that checking program equivalence via WP is difficult, if not impossible, since it amounts to reasoning about uncountably many postexpectations $g$.
We refer interested readers to~\cite[Chaps.\ 1--4]{barthe_katoen_silva_2020} for more recent advancements in formal semantics of probabilistic programs.

\paragraph*{Probabilistic Inference.}
There are a handful of probabilistic systems that employ an alternative forward semantics based on \emph{probability density function} (PDF) representations of distributions, e.g., ($\lambda$)PSI~\cite{DBLP:conf/cav/GehrMV16,DBLP:conf/pldi/GehrSV20}, AQUA~\cite{DBLP:conf/atva/HuangDM21}, Hakaru~\cite{DBLP:conf/padl/CaretteS16,DBLP:conf/flops/NarayananCRSZ16}, and the density compiler in~\cite{DBLP:journals/lmcs/BhatBGR17,DBLP:conf/popl/BhatAVG12}. These systems are dedicated to 
probabilistic inference for programs encoding continuous distributions (or joint discrete-continuous distributions). Reasoning about the underlying PDF representations, however, amounts to resolving complex integral expressions in order to answer inference queries, thus confining these techniques either to (semi-)numerical methods \cite{DBLP:conf/atva/HuangDM21,DBLP:conf/padl/CaretteS16,DBLP:conf/flops/NarayananCRSZ16,DBLP:journals/lmcs/BhatBGR17,DBLP:conf/popl/BhatAVG12} or exact methods yet limited to bounded looping bahaviors \cite{DBLP:conf/cav/GehrMV16,DBLP:conf/pldi/GehrSV20}. Apart from these inference systems, a recently developed language called Dice~\cite{DBLP:journals/pacmpl/HoltzenBM20} featuring exact inference for discrete probabilistic programs is also confined to statically bounded loops. The tool Mora~\cite{DBLP:conf/tacas/BartocciKS20,DBLP:conf/ictac/BartocciKS20} supports exact inference for various types of Bayesian networks, but relies on a restricted form of intermediate representation known as prob-solvable loops, whose behaviors can be expressed by a system of C-finite recurrences admitting closed-form solutions.


\paragraph*{Equivalence of Probabilistic Programs.}
Murawski and Ouaknine~\cite{DBLP:conf/concur/MurawskiO05} showed an \textsc{Exptime} decidability result for checking the equivalence of probabilistic programs over \emph{finite} data types by recasting the problem in terms of probabilistic finite automata~\cite{DBLP:journals/siamcomp/Tzeng92,DBLP:conf/cav/KieferMOWW11,DBLP:journals/iandc/ForejtJKW14}. Their techniques have been automated in the equivalence checker APEX~\cite{DBLP:conf/tacas/LegayMOW08}. Barthe et al.~\cite{DBLP:conf/lics/BartheJK20} proved a 2-\textsc{Exptime} decidability result for checking equivalence of \emph{straight-line} probabilistic programs (with deterministic inputs and no loops nor recursion) interpreted over all possible extensions of a finite field. Barthe et al.~\cite{DBLP:conf/popl/BartheGB09} developed a relational Hoare logic for probabilistic programs, which has been extensively used for, amongst others, proving program equivalence with applications in provable security and side-channel analysis.

The decidability result established in this paper is \emph{orthogonal} to the aforementioned results:
\begin{enumerate*}[label=(\roman*)]
	\item our decidability for checking $L \sim S$ applies to discrete probabilistic programs $L$ with \emph{unbounded} looping behaviors over a possibly \emph{infinite} state space; the specification $S$ -- though, admitting no loops -- encodes a possibly \emph{infinite-support} distribution; yet as a compromise,
	\item our decidability result is confined to $\pplname$ programs that necessarily terminate almost-surely on all inputs, and involve only distributions with rational closed-form PGF.
\end{enumerate*}



\section{Conclusion and Future Work}
\label{sec:conclusion}

We showed the decidability of -- and have presented a fully-automated technique to verifying -- whether a (possibly unbounded) probabilistic loop is equivalent to a loop-free specification program.
Future directions include determining the complexity of our decision problem; amending the method to continuous distributions using, e.g., \emph{characteristic functions}; extending the 
notion of probabilistic equivalence to 
probabilistic refinements; exploring PGF-based counterexample-guided synthesis of quantitative loop invariants (see~\ifcameraready\cite[Appx.~F.6]{DBLP:journals/corr/abs-2205-01449} \else\Cref{proof:loop-free-decidability} \fi for generating counterexamples); and tackling Bayesian inference. 



\ifblindreview%
\else%
\subsubsection{Acknowledgments.}
The authors thank Philipp Schr\"oer for providing support for his tool \textsc{Probably}\footnote{\faGithub~\url{https://github.com/Philipp15b/Probably}.} 
which forms the basis of our implementation.
\fi

%
%
%
\bibliographystyle{splncs04}
\bibliography{references}

\ifcameraready
\else
\allowdisplaybreaks 
\appendix
\clearpage\section{Extended Syntax of $\pplname$}
\label{sec:syntsugar}

The examples presented in this paper make use of a richer programming language syntax than what is defined in \Cref{fig:seq_loops}; the following table defines these constructs in terms of the core language.
\begin{table}[h]
    \centering
    \begin{adjustbox}{max width=\textwidth}
        \setlength{\tabcolsep}{1em}
        \begin{tabular}{l  l}
            \toprule
            Extended syntax &  Expressed through...\\
            \midrule
            $\incrasgn{\progvar{x}}{\progvar{y}}$ & $\incrasgn{\progvar{x}}{\iid{\dirac{1}}{\progvar{y}}}$\\[0.5mm]
            $\asgn{\progvar{x}}{\progvar{y}}$ & $\asgn{\progvar{x}}{0} \fatsemi \incrasgn{\progvar{x}}{\progvar{y}}$ \\
            \midrule
            $\ASSIGN{\progvar{x}}{D}$ & $\ASSIGN{\progvar{tmp}}{1}\fatsemi \ASSIGN{\progvar{x}}{\iid{D}{\progvar{tmp}}}\fatsemi\ASSIGN{\progvar{tmp}}{0}$\\[1mm]
            $\PCHOICE{P_1}{p}{P_2}$ & \makecell[l]{$\ASSIGN{\progvar{tmp}}{\bernoulli{p}}\fatsemi$\\ $\IF{\progvar{tmp} = 1}\{P_1\}~\ELSE\{P_2\}\fatsemi\ASSIGN{\progvar{tmp}}{0}$}\\
            \midrule
            $\IF{\varphi}\{P\}$ & $\IF{\varphi}\{P\}~ \ELSE \{\pskip\}$\\[0.5mm]
            $\IF{\neg \varphi}\{P\}$ & $\IF{\varphi}\{\pskip\}~\ELSE\{P\}$\\[0.5mm]
            $\IF{\varphi ~\wedge~ \psi}\{P\}$ & $\IF{\varphi}\{\IF{\psi}\{P\}\}$\\[0.5mm]
            $\IF{\varphi ~\vee~\psi}\{P\}$& $\IF{\neg (\neg \varphi ~\wedge~ \neg \psi)}\{P\}$\\[0.5mm]
            $\mathtt{switch}~ x~\{ \ldots~\mathtt{case}~n\colon P_n\fatsemi \mathtt{break}\fatsemi~ \ldots\}$ & $\ldots~\IF{x = n}\{P_n\}~\ELSE \IF{\ldots}\{\ldots\}$\\
            \midrule
            $\mathtt{repeat} ~n~ \mathtt{times} ~\{P\}$ & $\underbrace{P\fatsemi P\fatsemi P \ldots\fatsemi P}_{n~\text{times}}$\\
            \bottomrule\\
        \end{tabular}
    \end{adjustbox}

	\vspace{-2mm}
	\scriptsize{$^\star$We allow any Boolean structures in loop guards in a similar fashion as in the \pif~statements.}
\end{table}

\clearpage\section{Additional Remarks on Non-Rectangular Guards}
\label{app:nonrect}

Consider the following program $P_{1/\pi}$:

\begin{align*}
    & \asgn{\progvar{x}}{\geometric{1/4}} \,\fatsemi \\
    & \asgn{\progvar{y}}{\geometric{1/4}} \,\fatsemi \\
    & \asgn{\progvar{t}}{\progvar{x} + \progvar{y}} \,\fatsemi \\
    & \{\, \asgn{\progvar{t}}{\progvar{t} + 1} \,\} ~[5/9]~ \{\, \pskip \,\} \,\fatsemi \\
    & \asgn{\progvar{r}}{1} \,\fatsemi \\
    & \mathtt{repeat} ~3~ \mathtt{times} ~\{ \\
    & \qquad \asgn{\progvar{s}}{\iid{\bernoulli{1/2}}{2\progvar{t}}} \,\fatsemi \\
    & \qquad \IF{\progvar{s} \neq \progvar{t}} \{\, \asgn{\progvar{r}}{0} \,\} \\
    & \} 
\end{align*}

It was shown in~\cite{DBLP:conf/soda/FlajoletPS11} that after termination of $P_{1/\pi}$, variable $\progvar{r}$ is 1 with probability $1/\pi$.
Clearly, the program $P_{1/\pi}$ is $\pplname$ up to the $\mathtt{if}$-statement $\IF{\progvar{s} \neq \progvar{t}} \ldots$ which has a \emph{non-rectangular} guard.
Since $1/\pi$ is an irrational number, this example shows that non-rectangular guards do not preserve rational or algebraic closed forms.
Indeed, if the output distribution of $P_{1/\pi}$ had a rational (or algebraic) closed form $g$, then the probability that $\progvar{r} = 1$ after termination is given by
\[
    \subsFPSVarFor{    \partial_R (\subsFPSVarFor{g}{X, Y, T, S}{1}) }{R}{0}
\]
which is rational or algebraic, respectively.
\clearpage\section{Background on Basic Fixed Point Theory}
\label{app:fixpoint}

This paper fundamentally relies on two basic fixed point theoretic principles:
(i) Kleene's fixed point theorem and
(ii) a variant of a fixed point induction theorem that is sometimes referred to as \emph{Park induction}~\cite{park1969fixpoint}.

Recall that $(A, \sqsubseteq)$, where $A$ is a set and $\sqsubseteq$ is a partial order on $A$, is an $\omega$-complete partial order ($\omega$-cpo) if there exists a least element $\bot \in A$ and for all infinite ascending chains $a_0 \sqsubseteq a_1 \sqsubseteq \ldots$ there exists a least upper bound (supremum) $\sup_{i\geq 0} a_i \in A$.
Further, a function $f \colon A \to A$ is continuous if for all chains $a_0 \sqsubseteq a_1 \sqsubseteq \ldots$ it holds that $f(\sup_{i\geq 0} a_i) = \sup_{i \geq 0} f(a_i)$.
Continuity implies \emph{monotonicity}:
If $f$ is continuous, then for all $a, b \in A$ we have
\[
    a \ssqsubseteq b \qimplies f(a) \ssqsubseteq f(b) ~.
\]

The variants of Kleene's fixed point theorem and Park induction used in this paper are as follows:
\begin{theorem}
    \label{thm:kleenePark}
    Let $f \colon A \to A$ be a continuous function on $\omega$-cpo $(A, \sqsubseteq)$.
    Then
    \begin{itemize}
        \item (Kleene's fixed point theorem) $f$ has a unique least fixed point given by
        \[
            \lfp f \eeq \sup_{i\geq 0} f^i(\bot) ~.
        \]
        \item (Park induction) For all $x \in A$,
        \[
            f(x) \ssqsubseteq x \qimplies \lfp f \ssqsubseteq x ~.
        \]
    \end{itemize}
\end{theorem}
\begin{proof}
    We provide a self-contained proof.
    
    For Kleene's fixed point theorem note that
    \[
        f(\sup_{i \geq 0} f^i(\bot))
        \eeq
        \sup_{i \geq 0} f^{i+1}(\bot)
        \eeq
        \sup_{i \geq 0} f^{i}(\bot)
    \]
    where the first equality holds by continuity of $f$.
    Thus $\sup_{i \geq 0} f^i(\bot)$ is indeed a fixed point of $f$.
    To see that it is the least fixed point, let $x \in A$ such that $f(x) = x$.
    Since $\bot \sqsubseteq x$ we have $f(\bot) \sqsubseteq f(x) = x$ by monotonicity of $f$.
    Iterating this argument yields $f^i(\bot) \sqsubseteq x$ for all $i \geq 0$.
    This means that $x$ is an \emph{upper bound} on the chain $f^0(\bot) \sqsubseteq f^i(\bot) \sqsubseteq \ldots$, and thus it is greater or equal to the least upper bound (the supremum):
    \[
        \sup_{i \geq 0} f^i(\bot) \sqsubseteq x ~.
    \]
    This concludes the proof of Kleene's fixed point theorem.
    
    For Park induction assume that $x \in A$ is such that $f(x) \sqsubseteq x$.
    We claim that $f^i(\bot) \sqsubseteq x$ for all $i \geq 0$.
    This is proved by induction on $i$:
    \begin{itemize}
        \item For $i = 0$ this is trivial because $f^0(\bot) = \bot \sqsubseteq a$.
        \item For $i > 0$ we have $f^{i}(\bot) = f(f^{i-1}(\bot)) \overset{\text{I.H. + Mon.}}{\sqsubseteq} f(x) \overset{\text{Ass.}}{\sqsubseteq} x$.
    \end{itemize}
    Therefore, $x$ is again an upper bound on the chain $f^0(\bot) \sqsubseteq f^i(\bot) \sqsubseteq \ldots$ and thus
    \[
        x \ssqsupseteq \sup_{i \geq 0} f^i(\bot) \eeq \lfp f
    \]
    where the right equation is Kleene's fixed point theorem.
\end{proof}
\clearpage\section{Background on Formal Power Series}
\label{sec:fpsbackground}

This section is intended to be a self-contained introduction to FPS with a particular focus on the operations necessary for our technique.

Let $ k \in \N = \{0,1,\ldots\}$.
A $k$-dimensional \emph{formal power series} (FPS) is a map $f \colon \N^k \to \R$, i.e., a $k$-dimensional \emph{array} with real-valued entries.
It is common and useful to denote FPS as \enquote{formal sums}.
To this end, let $\vec{X} = (X_1,\ldots,X_k)$ be an ordered vector of formal symbols, called \emph{indeterminates}\footnote{Note that empty $\vec{X}$ is possible if $k =0$; a 0-dim.\ FPS is just a real number.}.
These indeterminates are merely \emph{labels} for the $k$ dimensions of $f$ and do not have any other particular meaning.
An FPS $f \colon \N^k \to \R$ is then written
\[
    f \eeq \fsum_{\sigma \in \N^k} f(\sigma) \vec{X}^\sigma
\]
where $\vec{X}^\sigma$ is the formal \emph{monomial} $X_1^{\sigma_1} X_2^{\sigma_2} \ldots X_n^{\sigma_n}$.
The exact \enquote{syntax} of the formal sum notation is not important and it is actually useful to allow a certain level of flexibility.
However, it is crucial that a formal sum unambiguously identifies each coefficient $f(\sigma)$ of the FPS.
For example, notations such as
\begin{equation}
    \label{eq:notationExamples}
    1 + X + Y, \qquad \fsum_{i \geq 5} X^i, \quad \text{or} \qquad X^5 + X^6 + X^7 + \ldots
\end{equation}
are all valid ways of denoting FPS (in fact, the latter two examples denote the same FPS).
The FPS where all coefficients are 0 is simply denoted $0$.
An FPS with $f(\sigma) \neq 0$ for finitely many $\sigma \in \N^k$ (like the leftmost example in \eqref{eq:notationExamples}) is called \emph{finite-support} or \emph{polynomial}, and otherwise it is called \emph{infinite-support}.
An FPS with $k=1$ ($k > 1$) is called \emph{univariate} (\emph{multivariate}, respectively).

\begin{remark}
    \label{rem:noFunction}
    One does not necessarily view FPS as \emph{analytic functions in their indeterminates}, even though the formal sum notation seems to suggest the opposite.
    This applies in particular to this paper for which there are several reasons:
    (1) For our approach, \emph{evaluating} the formal sums for arbitrary values of the indeterminates is neither meaningful nor necessary\footnote{Sometimes, we have to substitute an indeterminate by 0 or 1. Whenever we do so, we make sure that it is well-defined.}.
    (2) Relating an FPS to a function is problematic in general because the order in which the infinite sum is evaluated may affect the result; in fact, we would have to fix such an order upfront. However, especially for multivariate FPS, it is not clear what this order should be.
    (3) The formal, algebraic operations we apply to FPS are not always well-defined in the analytic sense.
    For instance, the FPS $\fsum_{i \geq 0} i!X^i$ converges for no $X \neq 0$~\cite{generatingfunctionology} and thus it is nowhere differentiable when viewed as an analytic function.
    However, its formal derivative is perfectly well-defined.
    
    The bottom line is that viewing FPS as purely algebraic objects, namely \emph{arrays} (rather than analytic functions) avoids many complications and is well-suited for the purposes of this paper.
\end{remark}

\subsubsection{Probability Generating Functions.}

A PGF $g$ is the special case of an FPS where the coefficients $g(\sigma)$ encode a (discrete) \emph{probability distribution} on $\N^k$:

\begin{definition}[PGF]
    \label{def:pgf}
    A $k$-dimensional FPS $g$ is called a probability generating function (PGF) if for all $\sigma \in \N^k$ we have \emph{(i)} $g(\sigma) \geq 0$ and \emph{(ii)} $\sum_{\sigma \in \N^k} g(\sigma) \leq 1$.
\end{definition}

For example,
\[
\frac{1}{3}X + \frac{1}{3}Y + \frac{1}{3}XY
\qqand
\frac{1}{2} + \frac{1}{4}X + \frac{1}{8}X^2 + \ldots
\]
are PGF; the former encodes a (finite-support) uniform distribution on the tuples $\{(1,0),(0,1),(1,1)\}$, and the latter is an (infinite-support) geometric distribution on the non-negative integers $\N$.

The sum in constraint (ii) of \Cref{def:pgf} is an \emph{actual infinite sum}, not a formal one.
However, this is not problematic because all coefficients are non-negative due to constraint (i), and thus we need not pay attention to issues of summation order.

In the following, we describe the relevant algebraic operations that can be performed on FPS and PGF in particular.

\subsubsection{Addition and multiplication.}
A key advantage of the sum notation is that it allows for a more intuitive definition of important FPS operations.
In fact, one can define \emph{sum} and \emph{product} operators such that the set $\fpsring{\vec{X}}$ of all $k$-dimensional FPS becomes a (commutative) ring, the \emph{ring of formal power series}.
These operations are defined as follows:
\begin{align*}
f + g \eeq \fsum_{\sigma \in \N^k} (f(\sigma) + g(\sigma))\vec{X}^\sigma
\qquad
fg \eeq \fsum_{\sigma \in \N^k} \sum_{\substack{ \tau + \rho = \sigma}} f(\tau) g(\rho) \vec{X}^\sigma
~.
\end{align*}
For the multiplication, the inner sum symbol is not a formal sum, but an actual \emph{finite} sum in the ring; it is finite because for every $\sigma \in \N^k$ there are just finitely many $\tau, \rho \in \N^k$ such that $\tau + \rho = \sigma$.

It is straightforward to verify that the two operations are associative, commutative and that multiplication distributes over addition.
Moreover, the constant FPS $0$ and $1$ are neutral w.r.t.\ to addition and multiplication, respectively, and thus $\fpsring{\vec{X}}$ is indeed a ring.

\subsubsection{Inverses.}
Even though $\fpsring{\vec{X}}$ is not a field for $k \geq 1$, there still exist elements $f$ with a unique \emph{inverse} $f^{-1} = 1/f$ satisfying $f f^{-1}= 1$ (uniqueness of inverses follows because $\fpsring{\vec{X}}$ is zero-divisor free).
Of particular interest for this paper (and in fact for the whole method of generating functions) are the cases where $f$ is an infinite-support FPS (meaning that $\{\sigma \in \N^k \mid f(\sigma) \neq 0\}$ is infinite) but $f^{-1}$ is a \emph{finite}-support FPS, also called \emph{polynomial}.
This enables representing the \emph{infinite} array $f$ implicitly through its inverse \emph{finite} array $f^{-1}$.
Well-known examples of this are the \emph{geometric series} and variants thereof, for example
\begin{equation}
\label{eq:inverseExample}
\frac{1}{2 - X}
\eeq
\fsum_{i \geq 0} \frac{1}{2^{i+1}}X^i ~.
\end{equation}
We stress that the \emph{only} meaning associated to equation \eqref{eq:inverseExample} is that $2 - X$ and $\fsum_{i \geq 0}1/2^{i+1} X^i$ are mutually inverse elements in the FPS ring.
In particular, we are usually \emph{not} concerned with viewing the left and right hand side as \enquote{functions in $X$}, cf.\ \Cref{rem:noFunction}.
We call a representation like $g = f/h = f h^{-1}$ a \emph{rational closed form} in this paper.
It is easy to show that for all $f,g,p,q \in \fpsring{\vec{X}}$ such that $g^{-1}$ and $q^{-1}$ exist it holds that, as expected,
\[
\frac{f}{g} \, \frac{p}{q} \eeq \frac{fp}{gq}
\qquad\text{and}\qquad
\frac{f}{g} + \frac{p}{q} \eeq \frac{fq + gp}{gq} ~,
\]
in other words, addition and multiplication preserve rational closed forms.

An alternative view on inverses is that for a given $f$ we look for a solution $x$ of the equation $fx = 1$ (in this case, the solution is unique if it exists).
In a similar way, we can define \emph{roots} in $\fpsring{\vec{X}}$: An $n$-th root of $f$ is a solution $x$ to the equation $x^n = f$.
In general, roots do not always exist and are not necessarily unique if they exist.
However, there are infinite-support FPS that have a closed form representation involving roots.
An example is the \enquote{Catalan distribution} that arises in the program in \Cref{fig:nested_loops}.
\twk{The explanation of roots should be improved}

\subsubsection{FPS as coefficients.}

For $k \geq 2$, a $k$-dimensional array can be viewed equivalently as an $l$-dimensional ($l < k$) array with $(k{-}l)$-dimensional arrays as entries, i.e., we identify $f \colon \N^k \to \R$ and $f' \colon \N^l \to (\N^{k-l} \to \R)$ if $f(\sigma) = f'(\tau)(\rho)$ for all $\sigma = (\tau, \rho)$.
In the formal sum notation, this is reflected by partitioning $\vec{X} = (\vec{Y}, \vec{Z}) = (Y_1,\ldots,Y_l,Z_{l+1},\ldots,Z_k)$ and viewing $f$ as an FPS in $\vec{Y}$ \emph{with coefficients that are FPS in the remaining indeterminates $\vec{Z}$}:
\[
\fsum_{\sigma \in \N^k} f(\sigma) \vec{X}^{\sigma}
\eeq
\fsum_{\tau \in \N^{l}} f'(\tau) \vec{Y}^\tau
~.
\]
Note that the formal sum notation on the right is not \enquote{problematic} as it still identifies all coefficients of $f$ in an unambiguous fashion.
In the following, we write $f(\tau)$ instead of $f'(\tau)$; the precise meaning of $f(\tau)$ for $\tau \in \N^l$, $l < k$, is to treat the FPS $f$ as a map of type $\N^l \to (\N^{k-l} \to \R)$ as explained above.
Partitioning the indeterminates in this manner is often useful for understanding the structure of an FPS.
For example, it holds that
\[
\fsum_{i,j \geq 0} \frac{1}{2^{2i + 1}} {i \choose j} Y^j X^i
\eeq
\fsum_{i \geq 0} \frac{1}{2^{i+1}} (\frac{1}{2} + \frac{1}{2} Y)^i X^i
\]
and the right hand side reveals almost immediately that this FPS is in fact the PGF resulting from executing the program
\begin{align*}
& \asgn{\progvar{x}}{\geometric{1/2}} \,\fatsemi \\
& \asgn{\progvar{y}}{\binomial{1/2}{\progvar{x}})} ~.
\end{align*}
Treating a given FPS as univariate in some distinguished indeterminate $X$ is particularly useful for defining the operations of \emph{substitution} and \emph{differentiation} that we explain in the following.

\subsubsection{Substitution.}
Substitution refers to the operation of replacing an indeterminate by another FPS.
Let $f = \fsum_{i \geq 0} f(i) X^i$ and $h$ be FPS.
Intuitively, the substitution of $X$ by $h$ in $f$ should satisfy the equality
\[
\subsFPSVarFor{f}{X}{h} \eeq \sum_{i \geq 0} f(i) h^i
\]
with the convention that $h^0 = 1$.
However, this equality cannot be seen as a general definition; recall from \Cref{rem:noFunction} that \enquote{evaluating} an FPS is often problematic, and substitution is even more general than evaluation.
Indeed, if $f = 1 + X + X^2 + \ldots = (1-X)^{-1}$, then neither $\subsFPSVarFor{f}{X}{1}$ nor $\subsFPSVarFor{f}{X}{1+X}$ are well-defined.
Examples of well-defined substitutions are $\subsFPSVarFor{f}{X}{0}$ and $\subsFPSVarFor{f}{X}{X^2}$, the former is equal to $1$ and the latter is the FPS $1 + X^2 + X^4 + \ldots = (1 - X^2)^{-1}$.
The next lemma states that substitution is well-defined in general if $h$ has no \enquote{constant coefficient} as in the last two examples:

\begin{lemma}
    Let $f = \fsum_{i \geq 0} f(i) X^i$ and $h$ be FPS in $\fpsring{X, \vec{Y}}$.
    If $h(0,\vec{0}) = 0$, then $\subsFPSVarFor{f}{X}{h}$ is a well-defined FPS and for all $i \geq 0$ it holds that
    \[
    \subsFPSVarFor{f}{X}{h}(i) \eeq \sum_{j \geq 0} f(j) h^j(i)
    \]
    where the right hand side is a \emph{finite} sum in $\fpsring{\vec{Y}}$.
\end{lemma}

It can be shown that such well-defined substitutions preserve the arithmetic operations, i.e., they are ring homomorphisms.
For example, it holds that
\[
fg \eeq 1 \qqimplies \subsFPSVarFor{f}{X}{h} \subsFPSVarFor{g}{X}{h} \eeq 1
\]
from which it follows that substitution preserves rational and other algebraic closed forms:
\[
    f \eeq \frac{p}{q} \qqimplies \subsFPSVarFor{f}{X}{h} \eeq \frac{\subsFPSVarFor{p}{X}{h}}{\subsFPSVarFor{q}{X}{h}} ~.
\]
For example, if we know that $(1-X)^{-1}$ is a closed form of $1 + X + X^2 + \ldots$, then $(1-X^2)^{-1}$ is a correct closed form of $1 + X^2 + X^4 + \ldots$.
\twk{include in future version: proof that substitution is homomorphism}


\subsubsection{Substitution by 1.}
Somewhat contrary to our \enquote{FPS aren't functions}-philosophy, we occasionally have to substitute an indeterminate by the constant $1$.
In the case of PGF, this corresponds to computing the marginal distribution in all \emph{other} variables -- a key operation necessary for executing assignments such as $\asgn{\progvar{x}}{5}$ (see \Cref{sec:pgfsem}).
However, as explained above, such substitutions are not well-defined for general FPS.
Therefore, we first identify a class of FPS where substitution by $1$ is well-defined: the \emph{$X$-absolutely convergent} FPS.

Let $f$ be an FPS in $\fpsring{X, \vec{Y}}$.
We say that $f$ is $X$-absolutely convergent if for all $\sigma \in \N^{\mathbf{Y}}$, the infinite series
\[
\sum_{j =0 }^\infty f(j, \sigma) 
\]
is absolutely convergent\footnote{A series $\sum_{i=0}^\infty a_i$ is called absolutely convergent if the sequence $(\sum_{i=0}^n |a_i|)_{n \geq 0}$ converges. Note that the sums in this notation are not \enquote{formal} as in the FPS notation; they refer to actual series and limits in $\mathbb{R}$.}.
If $f$ is $X$-absolutely convergent, then it makes good sense to define the substitution $\subsFPSVarFor{f}{X}{1}$ as $\subsFPSVarFor{f}{X}{1}(0,\sigma) = \sum_{j =0 }^\infty f(j, \sigma)$ and $\subsFPSVarFor{f}{X}{1}(i,\sigma) = 0$ for all $i > 0$ and $\sigma \in \N^{\mathbf{Y}}$.
\twk{in future version: show that the set of $X$-absolutely convergent FPS is a subring of $\fpsring{X, \vec{Y}}$}

For example, $f = \frac 1 2 + \frac 1 4 X + \frac 1 8 X^2 + \ldots$ and $g = 2-X$ are $X$-absolutely convergent, and we have $\subsFPSVarFor{f}{X}{1} = \subsFPSVarFor{g}{X}{1} = 1$.
Note that for finite-support (polynomial) FPS such as $g$, the substitution $\subsFPSVarFor{g}{X}{1}$ is an effective operation: simply replace all occurrences of $X$ in $g$ by $1$ and simplify the resulting term.

In general, all PGF are $X$-absolutely convergent for all their indeterminates $X$.
Second-order PGF (SOP, see \Cref{sec:loops}) are $X$-absolutely convergent for all non-meta indeterminates $X$; however, they are in general not $U$-absolutely convergent for meta indeterminates $U$.

The next lemma shows that substitution by $1$ can also be effectively performed on FPS for which a rational closed form is known:

\begin{lemma}
    \label{thm:subsByOne}
    Let $f$ and $g$ be FPS in $\fpsring{X, \vec{Y}}$ and suppose that both are $X$-absolutely convergent.
    Then,
    \[
    fg = 1 \qimplies \subsFPSVarFor{f}{X}{1} \subsFPSVarFor{g}{X}{1} = 1 ~.
    \]
\end{lemma}
\begin{proof}
    Intuitively, this observation is a straightforward application of the Cauchy product formula for absolutely converging power series.
    Formally, let $fg = 1$.
    Then by definition of the FPS product we have
    \begin{align}
    \label{eq:isone} f(0)g(0) & \eeq 1, \qand \\
    \label{eq:zeroterms} \sum_{i_1 + i_2 = i} f(i_1) g(i_2) & \eeq 0  \quad \text{for all } i>0 ~.
    \end{align}
    We now consider the product $\subsFPSVarFor{f}{X}{1} \subsFPSVarFor{g}{X}{1}$ and show that it is also equal to $1$:
    Let $\sigma \in \N^\mathbf{Y}$ be arbitrary.
    \begin{align*}
    & (\subsFPSVarFor{f}{X}{1} \subsFPSVarFor{g}{X}{1})(0, \sigma) \\
    \eeq & \sum_{\sigma_1 + \sigma_2 = \sigma}\subsFPSVarFor{f}{X}{1}(0, \sigma_1) \subsFPSVarFor{g}{X}{1}(0, \sigma_2) \tag{by def. of the FPS product}\\
    \eeq & \sum_{\sigma_1 + \sigma_2 = \sigma} \left(\sum_{j=0}^\infty f(j,\sigma_1)\right) \left(\sum_{j=0}^\infty g(j,\sigma_2) \right) \tag{by def. of substitution by $1$}\\
    \eeq & \sum_{\sigma_1 + \sigma_2 = \sigma} \sum_{j=0}^\infty \sum_{j_1 + j_2 = j}f(j_1, \sigma_1) g(j_2, \sigma_2) \tag{standard Cauchy product formula for absolutely converging series}\\
    \eeq & \sum_{\sigma_1 + \sigma_2 = \sigma} \left((f(0)g(0))(\sigma_1 + \sigma_2) +  \sum_{j=1}^\infty \sum_{j_1 + j_2 = j} (f(j_1) g(j_2))(\sigma_1 + \sigma_2)  \right)\tag{rewriting}\\
    \eeq & \sum_{\sigma_1 + \sigma_2 = \sigma} \left((f(0)g(0))(\sigma) +  \sum_{j=1}^\infty 0  \right)\tag{by \eqref{eq:zeroterms}}\\
    \eeq & \begin{cases}
    1 & \text{if } \sigma = \vec{0} \\
    0 & \text{else.}
    \end{cases}\tag{by \eqref{eq:isone}}\\
    \end{align*}
    Thus $\subsFPSVarFor{f}{X}{1} \subsFPSVarFor{g}{X}{1} = 1$.
\end{proof}

\Cref{thm:subsByOne} means that substitution of $X$ by $1$ also preserves inverses (provided that the involved FPS are $X$-absolutely convergent).
For example,
\begin{align*}
    \subsFPSVarFor{\frac{3}{2 X Y - 5 X - 4 Y + 10}}{X}{1}
    \eeq 
    \frac{3}{5 - 2Y}
    \eeq 
    \frac{0.6}{1 - 0.4Y}
    ~,
\end{align*}
implying that the left hand side is a PGF whose marginal in $Y$ is a geometric distribution with parameter $0.4$.

\twk{In future version say that subs by 1 is ring hom in abs converging subring}

Now let $f$ and $h$ be PGF.
We can also define substitution of $X$ in $f$ by the whole PGF $h$ (this is needed for the $\mathtt{iid}$-statement, see \Cref{sec:explainTransformer}).
Technically, we can recycle the two kinds of substitutions discussed above (substitution by $h$ with $h(0,\vec{0}) = 0$ and substitution by $1$):
\begin{corollary}
    Let $g$ and $h$ be PGF in $\fpsring{X, \vec{Y}}$.
    Let $h_T = h - h(0,\vec{0}) + h(0,\vec{0})T$ for a temporary auxiliary indeterminate $T$.
    Then
    \[
    \subsFPSVarFor{g}{X}{h} \ccoloneqq \subsFPSVarFor{\subsFPSVarFor{g}{X}{h_T}}{T}{1}
    \]
    is a well-defined PGF.
\end{corollary}
For the proof it suffices to note that since $h_T$ is still a PGF, $\subsFPSVarFor{g}{X}{h_T}$ is a PGF as well because PGF are closed under substitution.

\subsubsection{Derivatives.}
Differentiation is another basic algebraic operation on FPS.
Again, we consider a distinguished indeterminate $X$ and the FPS ring $\fpsring{X,\vec{Y}}$.
The formal (partial) derivative of an FPS $f = \fsum_{i \geq 0} f(i)X^i$ is defined as
\[
    \partial_X f \eeq \fsum_{i \geq 0} i \cdot f(i)X^{i-1} ~.
\]
Recall that the coefficients $f(i)$ are FPS in $\fpsring{\vec{Y}}$.
One application of derivatives in this paper is the extraction of these coefficients given some closed form of $f$:
\[
    f(i) \eeq \frac{1}{i!} \subsFPSVarFor{(\partial_X^i f)}{X}{0} ~.
\]
It can be shown that formal differentiation satisfies the familiar properties w.r.t.\ to sums, products and inverses, i.e.,
$\partial_X (f + g) = \partial_X f + \partial_X g$,
$\partial_X (fg) = \partial_X f + f \partial_X g$, and
$\partial_X \frac{1}{f} =  -\partial_X f / f^2$.
Notably, the derivative of an FPS expressed in a closed form is also in closed form.
This observation is crucial for our approach as it allows us to implement a closed form preserving PGF semantics for conditional branching.

\clearpage\section{Further Examples}
\label{apx:casestudies}

We list some additional examples which can be automatically verified using the our tool $\toolname$.
The timings are CPU-timings measured with the GiNaC backend without counting for parsing the input programs.
Even more examples (including features not discussed in this paper such as exact PGF-based Bayesian inference) can be found in the project files of the artifact.

\begin{example}[Geometric Distribution Generator]
\label{apx:exmp:geom}
This program generates a geometric distribution on $\progvar{c}$ using a while loop and fair coin flips.
The equivalence check using \texttt{GiNaC} as a backend was performed in 7.4ms.
\begin{figure}[h]
	\begin{adjustbox}{max width=\textwidth, max height=35mm}
		\begin{minipage}{0.6\textwidth}
			\begin{align*}
					& \WHILE{\progvar{x} = 1}\\
					& \qquad \PCHOICE{\ASSIGN{\progvar{x}}{0}}{1/2}{\ASSIGN{\progvar{c}}{\progvar{c} + 1}}\\
					&\}
			\end{align*}
		\end{minipage}
		\qquad
		\begin{minipage}{0.3\textwidth}
			\begin{align*}
				& \IF{\progvar{x} = 1}\{\\
				& \qquad \incrasgn{\progvar{c}}{\geometric{1/2}}\fatsemi\\
				& \qquad \ASSIGN{\progvar{x}}{0}\\
				& \}
			\end{align*}
		\end{minipage}
	\end{adjustbox}
	\label{fig:geometric_generator}
\end{figure}
\end{example}

\begin{example}[$n$-Geometric Distribution Generator]
\label{apx:exmp:ngeo}
This example is an equivalent version of the path generation program shown in the introduction.
$\toolname$ was able to automatically verify the invariant with the \texttt{GiNaC} backend in 4.8ms.
\begin{figure}[h]
	\begin{adjustbox}{max width=\textwidth, max height=35mm}
		\begin{minipage}{0.6\textwidth}
			\begin{align*}
				& \WHILE{\progvar{n} > 0} \\
				& \qquad \PCHOICE{\asgn{\progvar{n}}{\progvar{n} - 1}}{1/2}{\asgn{\progvar{c}}{\progvar{c} + 1}} \\
				& \}
			\end{align*}
		\end{minipage}
		\qquad
		\begin{minipage}{0.3\textwidth}
			\begin{align*}
				& \incrasgn{\progvar{c}}{\iid{\geometric{1/2}}{\progvar{n}}}\fatsemi\\
				& \ASSIGN{\progvar{n}}{0}
			\end{align*}
		\end{minipage}
	\end{adjustbox}
	\label{lab:kgeo}
\end{figure}
\end{example}

\begin{example}[IID-Sampling Statement]
\label{apx:exmp:iid_sampling}
This example illustrates the semantics of the $\mathtt{iid}$-statement (see \Cref{sec:iid}).
The example was automatically verified using the \texttt{GiNaC} backend in 14.5ms.
\begin{figure}[h]
	\centering
	\begin{adjustbox}{max width=\textwidth, max height=35mm}
		\begin{minipage}{0.6\textwidth}
			\begin{align*}
				& \WHILE{\progvar{n} > 0}\\
				&\qquad \incrasgn{\progvar{m}}{\mathtt{unif}(1,6)}\fatsemi\\
				&\qquad \ASSIGN{\progvar{n}}{\progvar{n} -1}\\
				&\}
			\end{align*}
		\end{minipage}
		\qquad
		\begin{minipage}{0.3\textwidth}
			\begin{align*}
				& \IF{\progvar{n} > 0}\{\\
				&\qquad \incrasgn{\progvar{m}}{\iid{\mathtt{unif}(1,6)}{\progvar{n}}}\fatsemi\\
				&\qquad \ASSIGN{\progvar{n}}{0}\\
				&\}
			\end{align*}
		\end{minipage}
	\end{adjustbox}
	\label{fig:iid_loops}
\end{figure}
\end{example}

\begin{example}[Random Walk]
\label{apx:exmp:random_walk}
This example implements a symmetric random walk on $\N$ which terminates upon reaching $\progvar{s} = 0$ for the first time.
It is well known that the program is UAST (even though the expected number of steps until termination is infinite regardless of the starting position).
We are able to verify that the runtime is distributed as the sum of $\progvar{s}$ i.i.d.\ samples from a linearly transformed Catalan distribution.
Using the \texttt{GiNaC} backend, the time to prove the equivalence was 9.8ms.
\begin{figure}[h]
	\centering
	\begin{adjustbox}{max width=\textwidth, max height=35mm}
		\begin{minipage}{0.6\textwidth}
			\begin{align*}
				& \WHILE{\progvar{s} > 0} \\
				& \qquad \PCHOICE{\ASSIGN{\progvar{s}}{\progvar{s}+1}}{1/2}{\ASSIGN{\progvar{s}}{\progvar{s} - 1}}\fatsemi\\
				& \qquad \ASSIGN{\progvar{c}}{\progvar{c} +1}\\
				&\}
			\end{align*}
		\end{minipage}
		\qquad
		\begin{minipage}{0.3\textwidth}
			\begin{align*}
				& \IF{\progvar{s} > 0} \{\\
				& \qquad \incrasgn{\progvar{c}}{\iid{2 \cdot \mathtt{catalan}(\sfrac{1}{2}) + 1}{\progvar{s}}}\fatsemi\\
				& \qquad \ASSIGN{\progvar{s}}{0}\\
				& \}
			\end{align*}
		\end{minipage}
	\end{adjustbox}
	\label{fig:random_walk}
\end{figure}
\end{example}

\begin{example}[Knuth-Yao Die]
	\label{apx:exmp:knuth_yao}
	The program in \Cref{fig:kydie} models the Knuth--Yao algorithm for simulating a fair six-sided die using coin flips only.
	$\pplname$ verifies the loop against its specification in 43.9ms using the \texttt{GiNaC} backend.
	\begin{figure}[h]
		\begin{adjustbox}{max width = \textwidth, max height=60mm}
			\begin{minipage}{0.4\textwidth}
				\begin{align*}
				& \WHILE{\progvar{s} < 7}\\
				& \qquad \mathtt{switch}~\progvar{s}~ \{\\
				& \qquad \qquad \pcase~ 0\colon \\
				& \qquad \qquad \qquad \PCHOICE{\ASSIGN{\progvar{s}}{1}}{1/2}{\ASSIGN{\progvar{s}}{2}}\fatsemi\\
				& \qquad \qquad \qquad \pbreak\fatsemi\\
				& \qquad \qquad \pcase~ 1\colon \\
				& \qquad \qquad \qquad \PCHOICE{\ASSIGN{\progvar{s}}{3}}{1/2}{\ASSIGN{\progvar{s}}{4}}\fatsemi\\
				& \qquad \qquad \qquad \pbreak\fatsemi\\
				& \qquad \qquad \pcase~ 2\colon \\
				& \qquad \qquad \qquad \PCHOICE{\ASSIGN{\progvar{s}}{5}}{1/2}{\ASSIGN{\progvar{s}}{6}}\fatsemi\\
				& \qquad \qquad \qquad \pbreak\fatsemi\\
				& \qquad \qquad \pcase~ 3\colon \\
				& \qquad \qquad \qquad \PCHOICE{\ASSIGN{\progvar{s}}{1}}{1/2}{\compose{\ASSIGN{\progvar{s}}{7}}{\ASSIGN{\progvar{die}}{1}}}\fatsemi\\
				& \qquad \qquad \qquad \pbreak\fatsemi\\
				& \qquad \qquad \pcase~ 4\colon \\
				& \qquad \qquad \qquad \PCHOICE{\compose{\ASSIGN{\progvar{s}}{7}}{\ASSIGN{\progvar{die}}{2}}}{1/2}{\compose{\ASSIGN{\progvar{s}}{7}}{\ASSIGN{\progvar{die}}{3}}}\fatsemi\\
				& \qquad \qquad \qquad \pbreak\fatsemi\\
				& \qquad \qquad \pcase~ 5\colon \\
				& \qquad \qquad \qquad \PCHOICE{\compose{\ASSIGN{\progvar{s}}{7}}{\ASSIGN{\progvar{die}}{4}}}{1/2}{\compose{\ASSIGN{\progvar{s}}{7}}{\ASSIGN{\progvar{die}}{5}}}\fatsemi\\
				& \qquad \qquad \qquad \pbreak\fatsemi\\
				& \qquad \qquad \pcase~ 6\colon \\
				& \qquad \qquad \qquad \PCHOICE{\ASSIGN{\progvar{s}}{2}}{1/2}{\compose{\ASSIGN{\progvar{s}}{7}}{\ASSIGN{\progvar{die}}{6}}}\fatsemi\\
				& \qquad \qquad \qquad \pbreak\fatsemi\\
				&\qquad \}\\
				& \}
				\end{align*}
			\end{minipage}
			\qquad
			\begin{minipage}{0.5\textwidth}
				\begin{align*}
				& \mathtt{switch}~ \progvar{s}~\{\\
				& \qquad \pcase~ 0\colon\\
				& \qquad \qquad \ASSIGN{\progvar{s}}{7}\fatsemi\\
				& \qquad \qquad \ASSIGN{\progvar{die}}{\mathtt{unif}(1,6)}\fatsemi\\
				& \qquad \qquad \pbreak\fatsemi\\
				& \qquad \pcase~ 1\colon\\
				& \qquad \qquad \ASSIGN{\progvar{s}}{7}\fatsemi\\
				& \qquad \qquad \ASSIGN{\progvar{die}}{\mathtt{unif}(1,3)}\fatsemi\\
				& \qquad \qquad \pbreak\fatsemi\\
				& \qquad \pcase~ 2\colon\\
				& \qquad \qquad \ASSIGN{\progvar{s}}{7}\fatsemi\\
				& \qquad \qquad \ASSIGN{\progvar{die}}{\mathtt{unif}(4,6)}\fatsemi\\
				& \qquad \qquad \pbreak\fatsemi\\
				& \qquad \pcase~ 3\colon\\
				& \qquad \qquad \ASSIGN{\progvar{s}}{7}\fatsemi\\
				& \qquad \qquad \PCHOICE{\ASSIGN{\progvar{die}}{1}}{2/3}{\PCHOICE{\ASSIGN{\progvar{die}}{2}}{1/2}{\ASSIGN{\progvar{die}}{3}}}\fatsemi\\
				& \qquad \qquad \pbreak\fatsemi\\
				& \qquad \pcase~ 4\colon\\
				& \qquad \qquad \ASSIGN{\progvar{s}}{7}\fatsemi\\
				& \qquad \qquad \PCHOICE{\ASSIGN{\progvar{die}}{2}}{1/2}{\ASSIGN{\progvar{die}}{3}}\fatsemi\\
				& \qquad \qquad \pbreak\fatsemi\\
				& \qquad \pcase~ 5\colon\\
				& \qquad \qquad \ASSIGN{\progvar{s}}{7}\fatsemi\\
				& \qquad \qquad \PCHOICE{\ASSIGN{\progvar{die}}{4}}{1/2}{\ASSIGN{\progvar{die}}{5}}\fatsemi\\
				& \qquad \qquad \pbreak\fatsemi\\
				& \qquad \pcase~ 6\colon\\
				& \qquad \qquad \ASSIGN{\progvar{s}}{7}\fatsemi\\
				& \qquad \qquad \PCHOICE{\ASSIGN{\progvar{die}}{6}}{2/3}{\PCHOICE{\ASSIGN{\progvar{die}}{4}}{1/2}{\ASSIGN{\progvar{die}}{5}}}\fatsemi\\
				& \qquad \qquad \pbreak\fatsemi\\
				&\}
				\end{align*}
			\end{minipage}
		\end{adjustbox}
		\caption{Knuth-Yao Die encoded as a \pgcl program together with its loop-free invariant.}
		\label{fig:kydie}
	\end{figure}
\end{example}

\begin{example}[Sequential Loops]
	\label{apx:exmp:seq_loops}
	Our technique also allows reasoning about \emph{sequential} loops by providing individual specifications for each occurring loop, as done in \Cref{fig:seq_loops}.
	The program presented is the sequential composition of \Cref{apx:exmp:iid_sampling} and \Cref{apx:exmp:random_walk}.
	\begin{figure}
		\begin{adjustbox}{max width=\textwidth, max height=35mm}
			\begin{minipage}[t]{0.45\textwidth}
				\begin{align*}
					& \WHILE{0 < \progvar{n}}\\
					& \qquad \incrasgn{\progvar{m}}{\mathtt{unif}(1,6)}\fatsemi\\
					& \qquad \decr{\progvar{n}}\\
					&\}\\
					& \WHILE{\progvar{m} > 0}\\
					& \qquad \PCHOICE{\incrasgn{\progvar{m}}{1}}{1/2}{\decr{\progvar{m}}}\fatsemi\\
					& \qquad \incrasgn{\progvar{c}}{1}\fatsemi\\
					&\}
				\end{align*}
			\end{minipage}
			\qquad
			\begin{minipage}[t]{0.45\textwidth}
				\begin{align*}
					& \comment{first invariant}\\
					& \IF{0 < \progvar{n}}\{\\
					& \qquad \incrasgn{\progvar{m}}{\iid{\mathtt{unif}(1,6)}{\progvar{n}}}\fatsemi\\
					& \qquad \ASSIGN{\progvar{n}}{0}\\
					&\} \\
                    & \comment{second invariant}\\
                    & \IF{\progvar{m} > 0} \{\\
                    & \qquad \incrasgn{\progvar{c}}{\iid{2 \cdot \mathtt{catalan}(\sfrac{1}{2}) + 1}{\progvar{m}}}\fatsemi\\
                    & \qquad \ASSIGN{\progvar{m}}{0}\\
                    & \}
				\end{align*}
			\end{minipage}
		\end{adjustbox}
	\caption{The two sequentially composed \while loops on the right are equivalent to the sequential composition of their specifications.}
	\label{fig:seq_loops}
	\end{figure}
\end{example}
\clearpage\section{Proofs Omitted in Main Text}

We do not present the proofs in the exact order in which they appear in the main text; instead, we first present the proof of \Cref{thm:soptrans} from which most other results follow as corollaries.

\subsection{Proof of \Cref{thm:soptrans}}
\label{proof:soptrans}

\noindent We restate \Cref{fig:synsem} for convenience.
\begin{table}[h]
    \centering
    \synsemtab
\end{table}

\restateSoptrans*

Recall that $\sem{P} \colon \sopdom \to \sopdom$ is defined inductively on the structure of $P$ according to the second column of \Cref{fig:synsem}.

We prove \Cref{thm:soptrans} with the help of an extra definition -- the notion of \emph{admissible} SOP-transformers -- and several lemmas.
For the whole proof, we fix $k$-many indeterminates $\vec{X} = (X_1,\ldots,X_k)$ corresponding to program variables and $l$-many \enquote{meta-indeterminates} $\vec{U} = (U_1,\ldots,U_l)$.
Therefore, in the whole section, the domains $\pgfdom$ and $\sopdom$ are subsets of $\fpsring{\vec{X}}$ and $\fpsring{\vec{U},\vec{X}}$, respectively.

\begin{definition}[Admissible SOP-transformer]
    \label{def:admissibleSOPTransformer}
    A function $\psi \colon \sopdom \to \sopdom$ is called \emph{admissible} if
    \begin{itemize}
        \item $\psi$ is \emph{continuous} on the $\omega$-cpo $\sopdom$.
        \item $\psi$ is \emph{linear} in the following sense:
        For all $f_1, f_2 \in \sopdom$ and $p \in [0,1]$
        \[
            pf_1 + f_2 \in \sopdom
            \qimplies
            \psi(pf_1 + f_2)
            \eeq
            p \psi(f_1) + \psi(f_2)
            ~.
        \]
        \item $\psi$ is \emph{homogeneous w.r.t.\ meta indeterminates}, i.e., for all $f \in \sopdom$ and $\tau \in \N^l$,
        \[
            \psi(f \vec{U}^\tau)
            \eeq
            \psi(f) \vec{U}^\tau
            ~.
        \]
        \item $\psi$ \emph{preserves PGF}, i.e., $g \in \pgfdom$ implies $\psi(g) \in \pgfdom$.
    \end{itemize}
\end{definition}

The general idea of the proof is to show that $\sem{P}$ is also an admissible SOP-transformer for every $\pplname$-program $P$.
To this end, we show that (i) the \enquote{atomic} FPS transformations like multiplication, substitution, etc., that are used in the definition of $\sem{P}$ are admissible, and (ii) admissible transformations are closed under composition and taking limits (the latter is needed for loops).
The claims from \Cref{thm:soptrans} then follow more or less directly from the admissibility of $\sem{P}$.

\begin{lemma}
    \label{thm:compAdmissible}
    Let $\psi_1$ and $\psi_2$ be admissible SOP-transformers.
    Their composition $\psi_1 \circ \psi_2$ is admissible as well.
\end{lemma}
\begin{proof}
    The proof is mostly standard and goes as follows:
    For continuity, let $f_0 \sqsubseteq f_1 \sqsubseteq \ldots$ be a chain in $\sopdom$.
    We have
    \[
        \sup_{i \geq 0} \psi_1(\psi_2(f_i))
        \eeq
        \psi_1(\sup_{i \geq 0} \psi_2(f_i))
        \eeq
        \psi_1(\psi_2(\sup_{i \geq 0}  f_i))
    \]
    where the first equation holds because $\psi_2(f_0) \sqsubseteq \psi_2(f_1) \sqsubseteq \ldots$ is also a chain in $\sopdom$ by monotonicity of $\psi_2$.
    
    Linearity of the composition follows because
    \[
        \psi_1(\psi_2(pf_1 + f_2))
        \eeq
        \psi_1(p\psi_2(f_1) + \psi_2(f_2)))
        \eeq
        p\psi_1(\psi_2(f_1)) + \psi_1(\psi_2(f_2)))
        ~.
    \]
    
    Homogeneity w.r.t.\ to the meta-indeterminates $\vec{U}$ holds because $\psi_1(\psi_2(f\vec{U}^\tau)) = \psi_1(\psi_2(f)\vec{U}^\tau) = \psi_1(\psi_2(f))\vec{U}^\tau$ for all $\tau \in \N^l$.
    
    PGF preservation is also trivial because both $\psi_1$ and $\psi_2$ preserve PGF.
\end{proof}

\begin{lemma}
    \label{thm:limitAdmissible}
    Let $\psi_0 \sqsubseteq \psi_1 \sqsubseteq \ldots $ be a chain in $(\sopdom \to \sopdom)$ and assume that $\psi_i$ is admissible for all $i \geq 0$.
    Then $\sup_{i \geq 0} \psi_i$ is admissible as well.
\end{lemma}
\begin{proof}
    It is clear that $\sup_{i \geq 0} \psi_i$ exists because $(\sopdom \to \sopdom)$ is an $\omega$-cpo.
    As a general observation, note that for all $f \in \sopdom$ it holds that
    \begin{align}
        \label{eq:obsBracketsSup}
        \sup_{i \geq 0} \psi_i (f) \eeq (\sup_{i \geq 0} \psi_i) (f)
    \end{align}
    due to the way the order is defined on $(\sopdom \to \sopdom)$.
    
    We now show that $\sup_{i \geq 0} \psi_i$ is continuous.
    To this end, let $f_0 \sqsubseteq f_1 \sqsubseteq \ldots$ be a chain in $\sopdom$.
    The following argument proves continuity:
    \begin{align*}
        & (\sup_{i \geq 0} \psi_i) (\sup_{j \geq 0} f_j) \\
        \eeq & \sup_{i \geq 0}\, \psi_i (\sup_{j \geq 0} f_j) \tag{By \eqref{eq:obsBracketsSup}}\\
        \eeq & \sup_{i \geq 0}\, \sup_{j \geq 0} \psi_i (f_j) \tag{$\psi_i$ is continuous by assumption}\\
        \eeq & \sup_{j \geq 0}\, \sup_{i \geq 0} \psi_i (f_j) \tag{suprema commute}\\
        \eeq & \sup_{j \geq 0}\, (\sup_{i \geq 0} \psi_i) (f_j) ~. \tag{By \eqref{eq:obsBracketsSup}}\\
    \end{align*}
    Next, we show that $\sup_{i \geq 0} \psi_i$ is linear.
    To this end, note that addition and \enquote{scalar multiplication} are continuous:
    For all chains $f_0 \sqsubseteq f_1 \sqsubseteq \ldots$ and in $\sopdom$ and $h \in \sopdom$ we have
    \begin{equation}
        \label{eq:addCont}
        \sup_{i\geq 0} (f_i + h) \eeq \sup_{i\geq 0} f_i + h
    \end{equation}
    because both addition of FPS and the order on $\sopdom$ are defined \enquote{coefficient-wise} and addition in $\R$ is continuous.
    For similar reasons, scalar multiplication is continuous, i.e., for all $p \in [0,1]$,
    \begin{equation}
        \label{eq:scalarMultCont}
        \sup_{i\geq 0} pf_i \eeq p \sup_{i\geq 0} f_i ~.
    \end{equation}
    Now let $p \in [0,1]$ and $f, g \in \sopdom$ such that $pg + f \in \sopdom$.
    Then
    \begin{align*}
        & (\sup_{i \geq 0} \psi_i)(pg + f) \\
        \eeq & \sup_{i \geq 0}\, \psi_i(pg + f) \tag{by \eqref{eq:obsBracketsSup}}\\
        \eeq & \sup_{i \geq 0}\, \left[p \psi_i(g) + \psi_i(f)\right] \tag{by assumption} \\
        \eeq & \sup_{i \geq 0}\, p \psi_i(g) + \sup_{i \geq 0}\, \psi_i(f) \tag{continuity of + as in \eqref{eq:addCont}} \\
        \eeq & p \sup_{i \geq 0}\, \psi_i(g) + \sup_{i \geq 0}\, \psi_i(f) \tag{continuity of scalar multiplication as in \eqref{eq:scalarMultCont}} \\
        \eeq & p (\sup_{i \geq 0} \psi_i)(g) + (\sup_{i \geq 0} \psi_i)(f) \tag{by \eqref{eq:obsBracketsSup}}
    \end{align*}
    which proves linearity of $\sup_{i \geq 0} \psi_i$.
    
    To show that $\sup_{i \geq 0} \psi_i$ is homogeneous w.r.t.\ meta-indeterminates we argue as follows:
    \begin{align*}
        & (\sup_{i \geq 0} \psi_i)(f \vec{U}^\tau) \\
        \eeq & \sup_{i \geq 0} \psi_i (f \vec{U}^\tau) \tag{By \eqref{eq:obsBracketsSup}}\\
        \eeq & \sup_{i \geq 0} \psi_i (f) \vec{U}^\tau \tag{$\psi_i$ is admissible}\\
        \eeq & (\sup_{i \geq 0} \psi_i) (f) \vec{U}^\tau \tag{By \eqref{eq:obsBracketsSup}} ~.
    \end{align*}
    
    Finally, for all $g \in \pgfdom$ it holds that $(\sup_{i \geq 0} \psi_i)(g) \in \pgfdom$ because $\sup_{i \geq 0} \psi_i(g)$ is an ascending chain in the $\omega$-cpo $\pgfdom$, and hence the supremum exists and is itself a PGF.
\end{proof}

\begin{lemma}
    \label{thm:sumadmissible}
    Let $\psi_1$ and $\psi_2$ be admissible SOP transformers, and let $\guard$ a guard of the form $\progvar{x} < n$.
    Then the transformer
    \[
        (\psi_1 +_{\guard} \psi_2) \colon \sopdom \to \sopdom,~ f \mmapsto \psi_1(\restrict{f}{\guard}) + \psi_2(\restrict{f}{\neg\guard})
    \]
    is admissible as well.
\end{lemma}
\begin{proof}
    Continuity is straightforward.
    Linearity follows easily because $\psi_1$ and $\psi_2$ are linear, and the same applies to homogeneity.
    Finally, PGF preservation applies because neither $\psi_1$ nor $\psi_2$ increase the probability mass of $\restrict{f}{\guard}$ and $\restrict{f}{\neg\guard}$ (for $f$ a PGF), respectively, and thus the mass of $\psi_1(\restrict{f}{\guard}) + \psi_2(\restrict{f}{\neg\guard})$ is at most one.
    Hence it is a PGF.
    \twk{todo: more detail}
\end{proof}

\begin{lemma}
    \label{thm:atomicOpsAdmissible}
    The following functions ($\bullet$ denotes the argument) are admissible:
    \begin{itemize}
        \item $g \bullet \colon \sopdom \to \sopdom$ with $g \in \pgfdom$ (multiplication by a constant PGF $g$),
    \end{itemize}
    and for a \underline{non-meta} indeterminate $X \in \vec{X}$,
    \begin{itemize}
        \item $\subsFPSVarFor{\bullet}{X}{g} \colon \sopdom \to \sopdom$ with $g \in \pgfdom$ (substitution by a constant PGF $g$).
    \end{itemize}
\end{lemma}
\begin{proof}
    In each of the following cases we show (i) continuity, (ii) linearity, (iii) homogeneity w.r.t.\ to meta invariants, and (iv) PGF preservation.
    \begin{itemize}
        \item \emph{Multiplication by a constant PGF $g$.}\\
        Let $f = \fsum_{\tau \in \N^l} f(\tau) \vec{U}^\tau$, where $f(\tau) \in \pgfdom$, be an arbitrary SOP.
        Note that $gf \in \sopdom$ because
        \[
            gf \eeq \fsum_{\tau \in \N^l} g f(\tau) \vec{U}^\tau
        \]
        and $g f(\tau) \in \pgfdom$ for all $\tau \in \N^l$ since PGF are closed under multiplication.

        (i) For continuity assume that $f_0 \sqsubseteq f_1 \sqsubseteq \ldots$ is a chain in $\sopdom$.
        \begin{align*}
        & \sup_{i \geq 0} g f_i \\
        \eeq & \sup_{i \geq 0} \fsum_{\tau \in \N^l} g f_i(\tau) \vec{U}^\tau \\
        \eeq & \fsum_{\tau \in \N^l} \sup_{i \geq 0} g f_i(\tau) \vec{U}^\tau \tag{order on $\sopdom$ is pointwise}\\
        \eeq & \fsum_{\tau \in \N^l} g \sup_{i \geq 0} f_i(\tau) \vec{U}^\tau \tag{see below}\\
        \eeq & g \fsum_{\tau \in \N^l} \sup_{i \geq 0} f_i(\tau) \vec{U}^\tau \\
        \eeq & g \sup_{i \geq 0} \fsum_{\tau \in \N^l} f_i(\tau) \vec{U}^\tau \\
        \eeq & g \sup_{i \geq 0} f_i ~,
        \end{align*}
        where we have used the fact that multiplication by $g$ is a continuous function in $\pgfdom$.
        The latter is proved as follows:
        Let $g_1 \sqsubseteq g_2 \sqsubseteq \ldots$ be a chain in $\pgfdom$.
        Then
        \begin{align*}
        & \sup_{i \geq 0} g g_i \\
        \eeq & \sup_{i \geq 0} \fsum_{\sigma \in \N^k} \sum_{\sigma_1 + \sigma_2 = \sigma} g(\sigma_1)g_i(\sigma_2) \\
        \eeq &  \fsum_{\sigma \in \N^k} \sup_{i \geq 0} \sum_{\sigma_1 + \sigma_2 = \sigma} g(\sigma_1)g_i(\sigma_2) \tag{because $\sqsubseteq$ is the pointwise order} \\
        \eeq &  \fsum_{\sigma \in \N^k}  \sum_{\sigma_1 + \sigma_2 = \sigma} g(\sigma_1) \sup_{i \geq 0} g_i(\sigma_2) \tag{addition and multiplication in $\R$ are continuous} \\
        \eeq &   g \sup_{i \geq 0} g_i
        \end{align*}
        
        (ii) Linearity follows directly from distributivity, commutativity and associativity of the FPS ring:
        \[
            g(pf_1 + f_2)
            \eeq
            p(gf_1) + gf_2
            ~.
        \]
        
        (iii) Homogeneity w.r.t.\ to meta indeterminates is just associativity:
        \[
            g(f \vec{U}^\tau) =  (gf) \vec{U}^\tau
        \]
        
        (iv) PGF preservation holds because PGF are closed under multiplication.
        \medskip \item \emph{Substitution by a constant PGF $g$.} \\
        \twk{in the future we should write down these boring proofs...}
        (i), (ii): Continuity and linearity can be shown similar as above.
        
        (iii) We consider an SOP $f \in \fpsring{X, \vec{Y}, \vec{U}}$ where $\vec{Y}$ are the non-meta indeterminates other than $X$.
        Homogeneity w.r.t.\ meta-indeterminates holds essentially because for all $i \geq 0$ we have $(f\vec{U}^\tau)(i) = f(i) \vec{U}^\tau$ as $X$ is not in $\vec{U}$.
        Thus
        \[
            \subsFPSVarFor{(f\vec{U}^\tau)}{X}{g}
            \eeq
            \sum_{i \geq 0} (f\vec{U}^\tau)(i) g^i
            \eeq
            \sum_{i \geq 0} f(i) g^i \vec{U}^\tau
            \eeq
            \subsFPSVarFor{f}{X}{g} \vec{U}^\tau
            ~.
        \]
        
        (iv) PGF preservation holds since PGF are closed under substitution; indeed, the PGF resulting from substituting an indeterminate $X$ in some PGF $f$ for another PGF $g$ describes the distribution of the sum of $n$ many i.i.d.\ samples from $g$, where $n$ is distributed as $X$ in $f$, cf.\ \Cref{sec:explainTransformer}.
        %
        %
        %
    \end{itemize}
    
\end{proof}

\begin{lemma}
    \label{thm:lemInfLin}
    An admissible SOP-transformer $\psi$ satisfies the following \enquote{infinite} linearity property:
    For all $f \in \sopdom$,
    \[
        \psi(f)
        \eeq
        \fsum_{\sigma \in \N^k, \tau \in \N^l} f(\sigma,\tau) \psi (\vec{X}^\sigma) \vec{U}^\tau
        ~.
    \]
\end{lemma}
\begin{proof}
    The proof uses continuity, \enquote{finite} linearity, and homogeneity w.r.t.\ $\vec{U}$.
    
    Let $f \in \sopdom$.
    Further, let $\pi_1 \colon \N \to \N^k$ and $\pi_2 \colon \N \to \N^l$ be such that $\pi \colon \N \to \N^{k+l}$, $\pi(i) = (\pi_1(i), \pi_2(i))$ is an arbitrary enumeration of $\N^{k+l}$.
    Then $\sum_{i=0}^n f(\pi(i)) \vec{X}^{\pi_1(i)}\vec{U}^{\pi_2(i)}$ is an ascending chain in $n$ with supremum $f$.
    Consequently,
    \begin{align*}
        & \psi(f) \\
        \eeq & \sup_{n \geq 0}\, \psi\left(\sum_{i=0}^n f(\pi(i)) \vec{X}^{\pi_1(i)}\vec{U}^{\pi_2(i)}\right) \tag{continuity of $\psi$} \\
        \eeq & \sup_{n \geq 0}\, \sum_{i=0}^n f(\pi(i)) \psi(\vec{X}^{\pi_1(i)}\vec{U}^{\pi_2(i)}) \tag{$n{-}1$ times \enquote{finite} linearity of $\psi$ }\\
        \eeq & \sup_{n \geq 0}\, \sum_{i=0}^n f(\pi(i)) \psi(\vec{X}^{\pi_1(i)} ) \vec{U}^{\pi_2(i)} \tag{homogeinity of $f$ w.r.t.\ meta-indeterminates $\vec{U}$}\\
        \eeq & \fsum_{\sigma \in \N^k, \tau \in \N^l} f(\sigma,\tau) \psi (\vec{X}^\sigma) \vec{U}^\tau ~.
    \end{align*}
    as claimed.
\end{proof}

\begin{lemma}
    \label{thm:lemPAdmissible}
    Let $P$ be an arbitrary $\pplname$-program over variables $\progvar{x}_1,\ldots,\progvar{x}_k$.
    The function $\sem{P}$ defined according to \Cref{fig:synsem} is an admissible SOP-transformer.
\end{lemma}
\begin{proof}
    If $P$ is loop-free, then the claim follows immediately from \Cref{thm:atomicOpsAdmissible}, \Cref{thm:compAdmissible} and \Cref{thm:sumadmissible} because $\sem{P}$ is just the composition of finitely many atomic admissible SOP-transformations (such as multiplication, substitution, etc.).
    
    Now suppose that $P = \WHILEDO{\guard}{P'}$ and recall that $\sem{P}$ is per definition equal to the lfp of the characteristic functional
    \[
        \charfun{\guard}{P'} \colon (\sopdom \to \sopdom) \to (\sopdom \to \sopdom),~
        \psi \mmapsto \lambda f .~ \restrict{f}{\neg\guard} \pplus \psi(\sem{P'}(\restrict{f}{\guard})) ~.
    \]
    A straightforward generalization of the proof in \Cref{proof:fpInduction} reveals that $\charfun{\guard}{P'}$ is a continuous function on the $\omega$-cpo $(\sopdom \to \sopdom)$.
    Continuity of $\charfun{\guard}{P'}$ implies the existence of $\lfp \charfun{\guard}{P'}$ by Kleene's fixed point theorem.
    In particular,
    \[
        \lfp \charfun{\guard}{P'}
        \eeq
        \sup_{i \geq 0} \charfun{\guard}{P'}^i(\bot) ~,
    \]
    where $\bot$ is the constant function that maps everything on the zero-SOP.
    By \Cref{thm:limitAdmissible}, $\lfp \charfun{\guard}{P'} = \sem{P}$ is admissible.
    
\end{proof}

\Cref{thm:lemInfLin} and \Cref{thm:lemPAdmissible} together imply \Cref{thm:soptrans}.


\subsection{Proof of \Cref{thm:fpInduction}}
\label{proof:fpInduction}

\restateFpInduction*

This result is an instance of Park induction (see \Cref{thm:kleenePark} in \Cref{app:fixpoint}):
We only have to show that the function
\[
    \charfun{\guard}{P'} \colon (\pgfdom \to \pgfdom) \to (\pgfdom \to \pgfdom),~ \psi \mmapsto \lambda g .~ \restrict{g}{\neg\guard} \pplus \psi(\sem{P'}(\restrict{g}{\guard}))
\]
is continuous on the $\omega$-cpo $(\pgfdom \to \pgfdom)$.
To show this, let $\psi_0 \sqsubseteq \psi_1 \sqsubseteq \ldots $ be a chain in $(\pgfdom \to \pgfdom)$.
As a general observation, note that, similar to \eqref{eq:obsBracketsSup}, we have for all $g \in \pgfdom$ that
\begin{align}
    \label{eq:swapLambdaSup}
    \sup_{i \geq 0} \psi_i
    \eeq
    \sup_{i \geq 0} [\lambda g .~ \psi_i(g)]
    \eeq
    \lambda g .~ [\sup_{i \geq 0} \psi_i(g) ]
\end{align}
by definition of the order on $(\pgfdom \to \pgfdom)$.
The following argument completes the proof:
\begin{align*}
    & \sup_{i \geq 0} \charfun{\guard}{P'}(\psi_i) \\
    \eeq & \sup_{i \geq 0} [\lambda g .~ \restrict{g}{\neg\guard} \pplus \psi_i(\sem{P'}(\restrict{g}{\guard}))] \tag{Def. of $\charfun{\guard}{P'}$} \\
    \eeq &  \lambda g .~ \sup_{i \geq 0} \left[\restrict{g}{\neg\guard} \pplus \psi_i(\sem{P'}(\restrict{g}{\guard}))\right] \tag{By \eqref{eq:swapLambdaSup}} \\
    \eeq &  \lambda g .~ \restrict{g}{\neg\guard} \pplus \sup_{i \geq 0}  \psi_i(\sem{P'}(\restrict{g}{\guard})) \tag{Addition is continuous, see \Cref{thm:atomicOpsAdmissible}} \\
    \eeq &  \lambda g .~ \restrict{g}{\neg\guard} \pplus (\sup_{i \geq 0}  \psi_i)(\sem{P'}(\restrict{g}{\guard})) \tag{By \eqref{eq:obsBracketsSup}} \\
    \eeq &  \charfun{\guard}{P'}(\sup_{i \geq 0} \psi_i) ~. \tag{Def. of $\charfun{\guard}{P'}$} \\
\end{align*}


\subsection{Proof of \Cref{thm:semPprops}}
\label{proof:semPProps}

\restateSemPProps*

\noindent We prove \Cref{thm:semPprops} as a simple consequence of the more general \Cref{thm:soptrans} (see \Cref{proof:soptrans}):

\begin{itemize}
    \item First, well-definedness of $\sem{P} \colon \pgfdom \to \pgfdom$ is clear because $\sem{P}$ is even well-defined as a more general SOP-transformer.
    \item Similarly, $\sem{P}$ is admissible on $\sopdom$ (\Cref{thm:lemPAdmissible}) and thus in particular continuous on $\pgfdom \subseteq \sopdom$.
    \item \enquote{Infinite} linearty of $\sem{P}$ on $\pgfdom$ follows directly from \Cref{thm:lemInfLin} (using again the fact that $\sem{P}$ is admissible by \Cref{thm:lemPAdmissible}).
\end{itemize}


\subsection{Proof of \Cref{thm:closedform}}

\Cref{thm:closedform} applies even to SOP:

\begin{theorem}
    Let $P$ be a \emph{loop-free} $\pplname$ program, and let $g = h/f \in \sopdom$ be in rational closed form.
    Then we can compute a rational closed form of $\sem{P}(g)  \in \sopdom$ by applying the transformations in \Cref{fig:synsem}.
\end{theorem}

This statement can be shown by induction over the structure of loop-free $\pplname$-programs.
The base cases ($\asgn{\progvar{x}}{0}$, $\decr{\progvar{x}}$, and $\incrasgn{\progvar{x}}{\iid{D}{\progvar{y}}}$) follow because they only need multiplication, addition, and substitution which preserve rational closed forms, see \Cref{sec:fpsbackground}.
Note that in particular, only non-meta indeterminates $X$ and $Y$ are substituted which is well-defined because an SOP is $Z$-absolutely convergent for all its non-meta indeterminates $Z$.
The other cases are also straightforward:
\begin{itemize}
    \item \emph{Sequential composition.} $\sem{\compose{P_1}{P_2}}(g) = \sem{P_2}(\sem{P_1}(g))$.
    By the IH we can assume that we have a rational closed form $g' = \sem{P_1}(g)$, and that $\sem{P_2}$ further transform this in a rational closed form $g'' = \sem{P_2}(g')$.
    \item \emph{Conditional choice.}
    A closed form for $g_{\progvar{x} < n}$ can be computed because differentiation preserves rational closed forms as well, see \Cref{sec:fpsbackground}.
    By the IH, both $\sem{P_1}(g_{\progvar{x} < n})$ as well as $\sem{P_2}(g-g_{\progvar{x} < n})$ have a rational closed form which we then simply add up.
\end{itemize}


\subsection{Proof of \Cref{thm:equivCharac}}
\label{proof:equivCharac}

\restateEquivCharac*

For the proof observe that 
\[
    g_{\vec{X}}
    \eeq
    \fsum_{\sigma \in \N^k} \vec{X}^\sigma \vec{U}^\sigma
    ~.
\]
The proof is as follows:
\begin{align*}
    & \sem{P_1}(g_{\vec{X}}) = \sem{P_2}(g_{\vec{X}}) \\
    \iff\quad & \sem{P_1}(g_{\vec{X}}) - \sem{P_2}(g_{\vec{X}}) = 0 \\
    \iff\quad & \sem{P_1}(\fsum_{\sigma \in \N^k} \vec{X}^\sigma \vec{U}^\sigma) - \sem{P_2}(\fsum_{\sigma \in \N^k} \vec{X}^\sigma \vec{U}^\sigma) = 0 \\
    \iff\quad & \fsum_{\sigma \in \N^k} \sem{P_1}(\vec{X}^\sigma) \vec{U}^\sigma - \fsum_{\sigma \in \N^k} \sem{P_2}(\vec{X}^\sigma) \vec{U}^\sigma = 0 \tag{By \Cref{thm:soptrans}} \\
    \iff\quad & \fsum_{\sigma \in \N^k} ( \sem{P_1}(\vec{X}^\sigma) -  \sem{P_2}(\vec{X}^\sigma)) \vec{U}^\sigma = 0 \tag{rewriting} \\
    \iff\quad & \forall\sigma \in \N^k \colon \sem{P_1}(\vec{X}^\sigma) -  \sem{P_2}(\vec{X}^\sigma) = 0 \tag{By definition of the 0-FPS in $\fpsring{\vec{X},\vec{U}}$}\\
    \iff\quad & \forall\sigma \in \N^k \colon \sem{P_1}(\vec{X}^\sigma) = \sem{P_2}(\vec{X}^\sigma) \\
    \iff\quad & \sem{P_1} = \sem{P_2} ~.
\end{align*}



\subsection{Proof of \Cref{thm:loop-free-decidability}}
\label{proof:loop-free-decidability}

\restateLoopFreeDecidability*

\noindent By \Cref{thm:equivCharac}, it suffices to decide whether $\sem{P_1}(g_{\vec{X}}) = \sem{P_2}(g_{\vec{X}})$.
Since $g_{\vec{X}}$ is an SOP expressible in the \emph{rational closed form}
\[
    g_{\vec{X}} \eeq \frac{1}{1 - X_1 U_1} \frac{1}{1 - X_2 U_2} \cdots \frac{1}{1 - X_k U_k} \iin \fpsring{\vec{U},\vec{X}} ~,
\]
it follows that both $\sem{P_1}(g_{\vec{X}})$ and $\sem{P_2}(g_{\vec{X}})$ have a rational closed form as well which is moreover effectively constructable by applying the rational-function preserving operations in \Cref{fig:synsem} (recall that $P_1$ and $P_2$ are both loop-free and sample only from distributions with rational PGF by assumption).

In $\fpsring{\vec{U},\vec{X}}$, it is easily decidable whether two FPS presented as rational closed forms $f_1 / h_1$ and $f_2 / h_2$ are equal:
\begin{align*}
    &\frac{f_1}{h_1} \eeq \frac{f_2}{h_2} \\
    \iff\quad & f_1 h_2 \eeq f_2 h_1 \tag{by FPS arithmetic} ~.
\end{align*}
Note that both sides on the latter equation are finite-support FPS -- \emph{polynomials} -- with rational coefficients.
Thus, we can simply compute these two polynomials and check whether their (finitely many) non-zero coefficients coincide.
If yes, then $P_1$ and $P_2$ are equivalent (i.e., $\sem{P_1} = \sem{P_2}$), and otherwise they are not equivalent.

\begin{remark}
    In the case $\sem{P_1} \neq \sem{P_2}$ we can also provide \emph{counterexamples} for inputs where the programs differ; to do so, we may simply expand one (or multiple) of the non-zero coefficients of $\sem{P_1} - \sem{P_2}$ in $\vec{U}$.
    For example, if $\vec{U} = (U_1,U_2)$ and we find that 
    \[
        \sem{P_1} - \sem{P_2} 
        \eeq 
        g U_1 U_2 + r
    \]
    where $g \neq 0$ is an FPS in $\fpsring{\vec{X}}$ and $r \in \fpsring{\vec{U},\vec{X}}$ not containing a term of the form $hU_1U_2$ ($h \in \fpsring{\vec{X}}$), then $P_1$ and $P_2$ differ on input $X_1 X_2$ (i.e., both $\progvar{x}_1$ and $\progvar{x}_2$ are initially 1 with probability 1), and the difference $\sem{P_1}(X_1X_2) - \sem{P_2}(X_1X_2)$ is given by $g$.
\end{remark}

\fi
\end{document}